\documentclass[]{interact} 
\usepackage{amsmath,amssymb}
\usepackage{graphicx}
\usepackage{color}
\usepackage{ulem}

\theoremstyle{plain} 
\newtheorem{thm}{Theorem}[section]
\newtheorem{cor}[thm]{Corollary}
\newtheorem{lem}[thm]{Lemma}
\newtheorem{prop}[thm]{Proposition}
\theoremstyle{definition} 
\newtheorem{defn}[thm]{Definition}
\theoremstyle{remark} 
\newtheorem{rem}[thm]{Remark}
\numberwithin{equation}{section}

\newcommand{\R}{\mathbb R}

\newcommand{\norm}[1]{\left\Vert#1\right\Vert}
\newcommand{\abs}[1]{\left\vert#1\right\vert}

\newcommand{\h}{\mathcal{H}}

\newcommand{\Z}{\mathbb Z}

\newcommand{\E}{\mathrm{e}}

\newcommand{\beq}{\begin{equation}}
\newcommand{\eeq}{\end{equation}}
\newcommand{\ben}{\begin{enumerate}}
\newcommand{\een}{\end{enumerate}}

\newcommand{\C}{\mathbb C}

\newcommand{\lan}{\langle}
\newcommand{\ran}{\rangle}

\newcommand{\iu}{\mathrm{i}\mkern1mu}

\begin{document}

\title{On balanced homodyne measurement \\
  -- simple proof of wave function collapse}

\author{
\name{E. Br\"uning\textsuperscript{a}\thanks{CONTACT E. Br\"uning. Email: bruninge@ukzn.ac.za} and S. Nagamachi\textsuperscript{b}}
\affil{\textsuperscript{a}School of Mathematics, Statistics and Computer Science,
	University of KwaZulu-Natal, Private Bag X54001, Durban 4000, and  and National Institute for Theoretical Physics (NITheP), KwaZulu-Natal, South Africa} 
 \textsuperscript{b}The University of Tokushima, Tokushima 770-8506, Japan}  
 
\maketitle

\begin{abstract}
Using only elementary calculus we prove that the balanced homodyne detector measures the quadrature phase amplitude of a signal.  More precisely, after the measurement of photon numbers $l$, the collapse of the composite state of a strong laser beam and a signal approximates the collapse of the signal to the eigen-state of quadrature phase amplitude with eigen-value $r$.	
\end{abstract}

\begin{keywords}
quantum teleportation, continuous variables, laser beam, homodyne detection;\\ 81V80,81P45,47N50
\end{keywords}

 \tableofcontents

\section{Introduction} \label{intro}
Optical homodyne detection is the extremely useful and flexible measuring method and has
broad applications in optical communication, e.g., quantum key distribution \cite{Ma21,Li21}, quantum teleportation \cite{Fu98}, quantum computation \cite{FT22} and weak signal detection \cite{Ge11,SF18}. In the balanced homodyne detection, the signal is mixed with a relative strong local oscillator on a half-beamsplitter, the two output modes are detected by a pair of detectors and the difference of two photo currents can be measured \cite{Br90, Ge11,SF18}. 

This article is a continuation of our articles \cite{NB08, BN19}.
In order to explain how balanced homodyne detection arises and how it is used in the theory of quantum teleportation we recall the relevant parts of \cite{NB08, BN19}:
{The total system is composed of three Hilbert spaces $\h _j$ ($j=0, 1, 2$) which are isomorphic each other.
$c_l^*$ and $c_j$ are the creation and the annihilation operators defined in $\h _j$ and $\vert n \ran _j = \frac{1}{\sqrt{n!}} c_j^{*n} \vert 0 \ran $ ($n=0, 1, \ldots $) constitutes an orthonormal basis of $\h _j$, where $\vert 0 \ran $ is the vacuum vector.  The system composed of $\h _1$ and $\h _2$ describes some localized part denoted part (A).  The Hilbert space $\h _1 \hat{\otimes} \h _2$ has a generalized basis
\beq \label{gbbasis}
  \left\{  \pi ^{-1/2} \sum _{n=0}^{\infty } (D_0(\alpha ) \vert n \ran _0) \otimes \vert n \ran _1 ; \alpha \in \C  \right\} , \ \ D_0(\alpha ) = \E ^{\alpha c_0^* - \bar{\alpha }c_0}
\eeq
called a generalized Bell basis, where $D_j(\alpha)$ denotes the operator of translation by $\alpha$ on $\h _j$.
The system $\h _3$ describes another localized part denoted by (B).  
}
  Alice {(a local observer at part (A))} performs the Bell-state measurements with respect to the generalized Bell basis
 on the state $\vert \Phi \ran $ of the composite system of $\vert \psi \ran _0 {= \sum _{n=0}^{\infty} a_n \vert n \ran _0}$ (to be teleported) and a quantum channel $\sum _{n=0}^{\infty } q^n \vert n \ran _1 \otimes \vert n \ran _2$ {(via which $\vert \psi \ran _0$ is teleported to Bob (a local observer at part (B)))}, i.e., on
$$
\vert \Phi \ran = \pi ^{1/2} \vert \psi \ran _0 \otimes \sum _{n=0}^{\infty } q^n \vert n \ran _1 \otimes \vert n \ran _2 \in \h _0 \hat{\otimes } \h _1 \hat{\otimes } \h _2, \ \abs{q} < 1 .
$$
The Bell-state measurement is done by sending the state $\vert \Phi \ran $ through the half-beamsplitter $\E ^{\iu H_{hbs}^{01}}$ and by the measurements of the observables $x_0$ and $p_1$:
\beq \label{in:01hbs}
H_{hbs}^{01} = \iu(\pi /4)(c_{0}^{*} c_{1} - c_{0}c_{1}^{*}), \ \  
x_0 = \frac{1}{\sqrt{2}} (c_0 + c_0^*), \  \ p_1 = \frac{1}{\sqrt{2} \iu } (c_1 - c_1^*).
\eeq
Let $\{ \vert x_- \ran _0; x_- \in \R \}$ be the generalized basis which consists of generalized eigen-vectors $\vert x_- \ran _0$ of $x_0$ with the eigen-value $x_-$, and $\{ \vert p_+ \ran _1; p_+ \in \R \}$ {the generalized basis consisting} of generalized eigen-vectors $\vert p_+ \ran _1$ of $p_1$ with eigen-value $p_+$.
Then we have (see Remark \ref{ho:hbs})
\beq \label{in:hbsxp}
\begin{array}{ll}
\E ^{-\iu H_{hbs}^{01}} \vert x_- \ran _0 \otimes \vert p_+ \ran _1 &= \pi ^{-1/2} \sum _{k=0}^{\infty } (D_0(x_- + \iu p_+ ) \vert k \ran _0) \otimes \vert k \ran _1 \\ [1mm]
\Leftrightarrow \ \ {}_0 \lan x_- \vert \otimes {}_1 \lan p_+ \vert  \E ^{\iu H_{hbs}^{01}} & = \pi ^{-1/2} \sum _{k=0}^{\infty } {}_0 \lan k \vert D_0(x_- + \iu p_+ )^* \otimes {}_1 \lan k \vert .
\end{array}
\eeq
The measurement theory called projective measurement is based on the following postulates (see \cite{Wm09}):
\begin{defn} [Projective measurement] \label{in:projmesur} 
	{\ } \newline \indent
	Postulate 1) When one measures the observable $\Lambda = \sum _{\lambda } \lambda \Pi ^{\lambda }$ on
	the state $\vert \psi \ran $, the result one obtains is one of the eigenvalue $\lambda $.  The probability for obtaining that particular eigenvalue $\lambda $ is
	$\lan \psi \vert \Pi ^{\lambda } \vert \psi \ran /\lan \psi \vert \psi \ran $,
	{where $\Pi ^{\lambda }$ is the projection appearing in the spectral resolution $\Lambda = \sum _{\lambda } \lambda \Pi ^{\lambda }$ of $\Lambda $.}
	
	Postulate 2) After the measurement the initial state $\vert \psi \ran $ changes to $\Pi ^{\lambda } \vert \psi \ran $.
\end{defn}
When Alice performs the projective measurement of the observable $x_0 \otimes p_1 \otimes I$ to the state $(\E ^{\iu H_{hbs}^{01}} \otimes I) \vert \Phi \ran $ and has the result $(x_0, p_1) = (x_-, p_+)$, then the state $(\E ^{\iu H_{hbs}^{01}} \otimes I) \vert \Phi \ran $ changes to 
$$
\begin{array}{ll}
{}&(\vert x_- \ran _0 \otimes \vert p_+ \ran _1 {}_0 \lan x_- \vert \otimes {}_1 \lan p_+ \vert \otimes I)(\E ^{\iu H_{hbs}^{01}} \otimes I \vert \Phi \ran )
\\ [1mm]
{}&= \vert x_- \ran _0 \otimes \vert p_+ \ran _1
\pi ^{-1/2} \sum _{k=0}^{\infty } {}_0 \lan k \vert D_0(x_- + \iu p_+ )^* \otimes {}_1 \lan k \vert \otimes I \vert \Phi \ran .
\end{array}
$$
The essential part of (the equation describing) teleportation is
\beq \label{pbmeasure}
\begin{array}{ll}
{}&  \pi ^{-1/2} \sum _{k=0}^{\infty } {}_0 \lan k \vert D_0(\alpha )^* \otimes {}_1 \lan k \vert \otimes I \vert \Phi \ran   =\\ [2mm]
{}& \sum _{k=0}^{\infty } {}_0 \lan k \vert D_0(\alpha )^* \vert \psi \ran _0 {}_1 \lan k \vert \sum _{n=0}^{\infty } q^n \vert n \ran _1 \otimes \vert n \ran _2 \\ [2mm]
{}&= \sum _{k=0}^{\infty } q^k {}_0 \lan k \vert D_0(\alpha )^* \vert \psi \ran _0 \vert k \ran _2 
\rightarrow D_2(\alpha )^* \vert \psi \ran _2
\end{array}
\eeq
as $q \rightarrow 1$, where $\vert \psi \ran _2 = \sum _{n=0}^{\infty } a_n \vert n \ran _2$.  Note that $\h _j$ are isomorphic to each other by the correspondence $\h _i \ni \vert n \ran _i \leftrightarrow \vert n \ran _j \in \h _j$.  This is the scheme of quantum teleportation of continuous variables (see \cite{Fu98, NB08, BN19}). 
 The postulate 2 of projective measuement plays an important role for teleportation.
\begin{rem} \label{in:rem1}
Sometimes the Bell-measurement is understood as 
the measurement of $\E ^{-\iu H_{hbs}^{01}} (x_0 \otimes p_1) \E ^{\iu H_{hbs}^{01}} \otimes I$ on the state $\vert \Phi \ran $ instead of the measurement
	$x_0 \otimes p_1 \otimes I$ on the state $(\E ^{\iu H_{hbs}^{01}} \otimes I) \vert \Phi \ran $.
	
	Postulate 1 gives the same result for both cases, but the postulate 2 gives different results.  In any case,  $\vert \Phi \ran $ and $(\E ^{\iu H_{hbs}^{01}} \otimes I) \vert \Phi \ran $ change to	
	$$
	\E ^{-\iu H_{hbs}^{01}} (\vert x_- \ran _0 \otimes \vert p_+ \ran _1) \otimes D_0(\alpha )^* \vert \psi \ran _2 \ \ 
	{\rm and} \ \  \vert x_- \ran _0 \otimes \vert p_+ \ran _1 \otimes D_0(\alpha )^* \vert \psi \ran _2
	$$
	respectively, that is, $\vert \psi \ran _0$ is teleported to $D_2(\alpha )^* \vert \psi \ran _2$.
\end{rem}
For the measurement of $x_{0}$, the balanced homodyne detection is often used \cite{Br90, Fu98, Ge11, SF18}.  

For the balanced homodyne detection, we prepare one more creation $c_3^*$ and annihilation $c_3$ operatores of a local oscillator and a half-beamsplitter
$$
  \E ^{\iu H_{hbs}^{03}}, \  H_{hbs}^{03} = \iu(\pi /4)(c_{0}^{*} c_{3} - c_{0}c_{3}^{*}) . 
$$
Let $\vert \alpha \ran _3$ be the coherent state 
$$
\vert \alpha \ran _3 = \E^{- \abs{\alpha }^2/2} \E^{\alpha c_{3}^{*}} \vert 0 \ran , \  \alpha \in \C ,
$$
where $\vert 0 \ran $ is the vacuum state.  For the measurement of $x_0$ we choose $\alpha = \abs{\alpha }$ and let $N_j$ ($j = 0, 3$) be the number operators defined by
$$
N_j = c_j^*c_j.
$$
Then $\vert n \ran _j = \frac{1}{\sqrt{n!}} c_j^{*n} \vert 0 \ran $ is an eigen-state of $N_j$ with eigen-value $n$, and
$$
\vert l + j \ran _3 \otimes \vert j \ran _0 \ (l \geq 0), \ \ \vert j \ran _3 \otimes \vert j - l \ran _0 \ (l \leq 0), \ \ j = 0, 1,  \ldots 
$$
are eigen-vectors of $N_3 - N_0$ with eigen-value $l$.

The balanced homodyne detection is the measurement of $N_3 - N_0$ on the state $e^{\iu H^{03}_{\rm hbs}}(\vert \phi \ran _0 \otimes \vert \alpha \ran _3)$, where $\vert \phi \ran _0$ is the state defined by
$$
\vert \phi \ran _0 = \phi (c_0^*) \vert 0 \ran 
$$
for an entire function $\phi (z)$ of $z$ with ${}_0 \lan \phi \vert \phi \ran _0 < \infty $.  The expectation value of this measurement is 
$$
{}_0 \lan \phi \vert \otimes {}_3 \lan \alpha \vert e^{-\iu H^{03}_{\rm hbs}} (N_3 - N_0) e^{\iu H^{03}_{\rm hbs}} \vert \phi \ran _0 \otimes \vert \alpha \ran _3
$$
which is the expectation value of 
$$
e^{-\iu H^{03}_{\rm hbs}} (N_3 - N_0) e^{\iu H^{03}_{\rm hbs}} = c_3^*c_0 + c_0^*c_3
$$
on the state $\vert \phi \ran _0 \otimes \vert \alpha \ran _3$.

Let {$\tilde{\Pi}^l$ and $\Pi^l$ be projections which appear in the spectral resolutions of $N_3 - N_0$ and $c_3^*c_0 + c_0^*c_3$ respectively, i.e.,} 
$$
N_3 - N_0 = \sum _{l = -\infty }^{\infty } l \tilde{\Pi }^l, \ c_3^*c_0 + c_0^*c_3 = \sum _{l = -\infty }^{\infty } l e^{-\iu H^{03}_{\rm hbs}} \tilde{\Pi }^l e^{\iu H^{03}_{\rm hbs}}
= \sum _{l = -\infty }^{\infty } l \Pi ^l .
$$
\begin{rem} \label{in:ncc*}
Under the projective measurement postulate, there are no differences of the measured value $l$ and the probability for obtaining $l$, between the measurement of $N_3 - N_0$ on the state $e^{\iu H^{03}_{\rm hbs}}(\vert \phi \ran _0 \otimes \vert \alpha \ran _3)$ and the measurement of $c_3^*c_0 + c_0^*c_3$ on the state $\vert \phi \ran _0 \otimes \vert \alpha \ran _3$, but the changed states are different, i.e., $\tilde{\Pi }^l e^{\iu H^{03}_{\rm hbs}}(\vert \phi \ran _0 \otimes \vert \alpha \ran _3)$ and $e^{-\iu H^{03}_{\rm hbs}} \tilde{\Pi }^l e^{\iu H^{03}_{\rm hbs}}(\vert \phi \ran _0 \otimes \vert \alpha \ran _3)$ (see (Remark \ref{in:rem1})).
\end{rem}
Note that one can deduce this
result from {the latter,} the result of projective measurement of the observable $c_3^*c_0 + c_0^*c_3$.  So, we study the measurement of the observable $c_3^*c_0 + c_0^*c_3$ on the state $\vert \phi \ran _0 \otimes \vert \alpha \ran _3$ under the framework of projective measurement.

The expectation value of the observable $c_3^*c_0 + c_0^*c_3$ on the state $\vert \phi \ran _0 \otimes \vert \alpha \ran _3$ is calculated as
$$
{}_0 \lan \phi \vert \otimes {}_3 \lan \alpha \vert (c_3^*c_0 + c_0^*c_3)\vert \phi \ran _0 \otimes \vert \alpha \ran _3
= {}_0 \lan \phi \vert (\bar{\alpha } c_0 + \alpha c_0^*) \vert \phi \ran _0 .
$$
The above formula is sometimes expressed as saying 
\beq \label{in:saying}
\begin{split}
	{\rm ``the\ local\ oscillator\ is\ strong,\ so\ that\ the\ operators\ } c_3 \ {\rm and\ } c_3^* \\
	{\rm can\ be\ replaced\ }
	{\rm by\ complex\ numbers\ } \alpha \ {\rm and\ } \bar{\alpha } \ {\rm respectively"}. 
\end{split}
\eeq
One might find the assumption (\ref{in:saying}) strange because of the following: Let $\alpha = \abs{\alpha }$.  Since the difference of the number operators haves eigen values in $\Z $, the statement  (\ref{in:saying}) seems to cause the following contradiction:
\beq \label{in:spec}
\begin{split}
	{\rm the\ operator\ } c_3^*c_0 + c_0^*c_3 \ {\rm has\ eigen\ values\ in\ } \Z ,  {\rm while\ the\ operator\ } \\ \bar{\alpha } c_0 + \alpha c_0^* = \abs{\alpha } (c_0 + c_0^*) = \sqrt{2} \abs{\alpha } x_0 \ 
	\ {\rm has\ its\ spectrum\ in\ } \R .
\end{split}
\eeq
Since the balanced homodyne detection seems to detect the number of photons, the eigen-values of $c_3^*c_0 + c_0^*c_3$, it seems impossible to make a measurement of $x_0$ by the balanced homodyne detection. 

But in \cite{Br90} it is
proven that the statement (\ref{in:saying}) is valid for the balanced homodyne measurement
in the sense that the probability of the outcome $l = x \abs{\alpha }$ for $c_3^*c_0 + c_0^*c_3$ and $\vert \phi \ran _0 \otimes \vert \alpha \ran _3$ approximates the probability of the outcome $x$ for $(1/\abs{\alpha })(\bar{\alpha } c_0 + \alpha c_0^*)$ and $\vert \phi \ran _0 $ for large $\abs{\alpha }$. 

It is convenient to use $(1/\abs{\alpha })(c_3^*c_0 + c_0^*c_3)$ and $(1/\abs{\alpha })(\bar{\alpha } c_0 + \alpha c_0^*)$ instead of $c_3^*c_0 + c_0^*c_3$ and $\bar{\alpha } c_0 + \alpha c_0^*$ for large $\abs{\alpha } \rightarrow \infty $.

Let $\alpha = \abs{\alpha } \E ^{\iu \theta }$ and $\vert \theta ; r\ran $ be the generalized eigen-vector of $\E^{-\iu \theta } c_0 + \E^{\iu \theta }c_0^*$ with eigen-value $r$ such that
$$
\frac{1}{\sqrt{2 \pi }} \int _{-\infty }^{\infty } \vert \theta ; r \ran _0
{}_0 \lan \theta ; r \vert  dr = I.
$$
The spectral resolution of $(1/\abs{\alpha })(c_3^*c_0 + c_0^*c_3)$ is $\sum _{l = - \infty }^{\infty } (l/\abs{\alpha }) \Pi ^l$ and the projection onto the space spanned by the eigen-states whose eigen-values $l/\abs{\alpha }$ are less than $x$ is
$$
\sum _{-\infty < l \leq x\abs{\alpha }}  \Pi ^l 
= \int _{- \infty }^{x \abs{\alpha }} \Pi ^{[s]}  ds
= \int _{- \infty}^x \abs{\alpha }  \Pi ^{[r \abs{\alpha }]}  dr ,
$$
which corresponds to 
$$
\frac{1}{\sqrt{2 \pi }} \int _{-\infty }^x \vert \theta ; r \ran _0
{}_0 \lan \theta ; r \vert  dr,
$$
where $[z]$ denotes the integer part of $z$.

In \cite{Br90}, Braunstein calculated the probability distribution $P(x)$ of the outcome $x = l/\abs{\alpha }$ for $(c_3^*c_0 + c_0^*c_3)/\abs{\alpha }$ and $\vert \beta \ran _0 \otimes \vert \alpha \ran _3$ by calculating the characteristic function $\chi (k) = ({}_0 \lan \beta \vert \otimes {}_3 \lan \alpha \vert )\exp \{ \iu k(c_3^*c_0 + c_0^*c_3)/\abs{\alpha }\} (\vert \beta \ran _0 \otimes \vert \alpha \ran _3)$ and its Fourier transformation, and he showed that 
$P(x)$ approximates 
\beq \label{in:braun}
  \frac{1}{\sqrt{2 \pi }} \E ^{-\{ x - (\bar{\alpha } \beta  + \alpha \bar{\beta })/\abs{\alpha }\} ^2/2}, \ 
  x = l/\abs{\alpha }
\eeq
well as $\abs{\alpha } \rightarrow \infty $.  It is easily seen that $\frac{1}{\sqrt{2 \pi }} \E ^{-\{ x - (\bar{\alpha } \beta  + \alpha \bar{\beta })/\abs{\alpha }\} ^2/2} = {}_0 \lan \beta \vert \theta ; x \ran _0
{}_0 \lan \theta ; x \vert \beta \ran _0 $ by using the expression (\ref{ho:eigenvec}), where $(\bar{\alpha } c_0 + \alpha c_0^*)/\abs{\alpha } = \int _{-\infty }^{\infty } x \vert \theta ; x \ran _0
{}_0 \lan \theta ; x \vert dx$ is the spectral resolution of $(\bar{\alpha } c_0 + \alpha c_0^*)/\abs{\alpha }$ for $\alpha = \E ^{\iu \theta } \abs{\alpha }$.  Thus the balanced homodyne detection fulfill the postulate 1 for the projective measurement.

However, in \cite{Br90} Braunstein did not study 
the postulate 2 of projective measurement, that is, does the collapse of $\vert \phi \ran _0 \otimes \vert \alpha \ran _3$ to an eigen-vector with eigen-value $l$ approximate the collapse of $\vert \phi \ran _0$ to an eigen-vector with eigen-value $x$?
The postulate 2 of projective measurement plays a crucial role for the theory of quantum teleportation.  We will prove the convergence of $\abs{\alpha }{}_3 \lan \alpha \vert \Pi ^{[x \abs{\alpha }]} \vert \alpha \ran _3 \rightarrow \vert \theta ; x \ran _0 {}_0 \lan \theta ; x \vert $ as $\abs{\alpha } \rightarrow \infty $, and give the affirmative answer to this question.
Our proof gives also the approximation of the probability distribution.  That is, our proof is an alternative proof of that given in \cite{Br90}, and this proof is elementary and direct (not via characteristic function).

In order to study the change of the state, we want to show that
\beq \label{in:strongcon}
\abs{\alpha } \Pi ^l (\vert \phi \ran _0 \otimes \vert \alpha \ran _3) \rightarrow  \frac{1}{\sqrt{2 \pi }} \vert \theta ; x \ran _0 
{}_0 \lan \theta ; x \vert \phi \ran _0 \otimes \vert \alpha \ran _3 
\eeq
as $\abs{\alpha } \rightarrow \infty $ for $l = [x \abs{\alpha }]$, where $[z]$ is the integer part of $z$.

In Section \ref{holomorph}, we give a brief review of the holomorphic representation of creation $c_j$ and annihilation $c_j^*$ ($j=0, 1, \ldots , n$) operators, where the coherent states $\vert \alpha \ran _2$ and $\vert \theta ; x \ran _1 {}_1 \lan \theta ; x \vert $ are represented by
$$
\E^{- \abs{\alpha }^2/2} \E^{\alpha \bar{u}_2} \ \ {\rm and} \ \ 
\E ^{x \E ^{\iu \theta } \bar{u}_1} \E ^{ - \E ^{2 \iu \theta } \bar{u}_1^2/2} \E ^{-x^2/4} 
\overline{\E ^{x \E ^{\iu \theta } \bar{v}_1}\E ^{ - \E ^{2 \iu \theta } \bar{v}_1^2/2} \E ^{-x^2/4}}
$$ 
respecctively.

In Sections \ref{balance} and  \ref{direct}, we use only $c_0$ and $c_3$.  So we change the subscripts $0, 3$ to $1, 2$, i.e., we replace $c_0$ and $c_3$ with $c_1$ and $c_2$ respectively.

In Section \ref{balance}, we introduce a set of vectors
$$
\vert m, n \ran  = \frac{1}{\sqrt{2^m 2^n m! n!}} (\bar{u}_1 + \bar{u}_2)^m(- \bar{u}_1 + \bar{u}_2)^n, \ \ 0 \leq m, n \in {\mathbb Z} 
$$
which constitutes an orthonormal basis of the Fock space $\h _1 \hat{\otimes } \h _2$ of $c_1$ and $c_2$, and show that 
the orthogonal projection $\Pi ^l$ onto the eigen space of $c_2^*c_1 + c_1^*c_2$ with eigen value $l$ is 
$$
\Pi ^l = \sum _{j = 0}^{\infty } \vert l + j, j \ran \lan l + j, j \vert \ \ ({\rm for } \ \ l \geq 0 ), \ \
\Pi ^l = \sum _{j = 0}^{\infty } \vert j, j - l \ran \lan j, j - l \vert \ \ ({\rm for } \ \ l \leq 0 ).
$$

Moreover, we introduce generalized vectors
$$
\vert \varphi , l \ran  = \sum _{j=0}^{\infty } \E ^{\iu j \varphi } \vert l+j, j \ran \ \ ({\rm for } \ \ l \geq 0 ), \ \ 
\vert \varphi , l \ran  = \sum _{j=0}^{\infty } \E ^{\iu j \varphi } \vert l, j - l \ran \ \ ({\rm for } \ \ l \leq 0 )
$$
and show
$$
\frac{1}{2 \pi } \int _{- \pi}^{\pi } \vert \varphi , l \ran \lan \varphi , l \vert d \varphi 
= \Pi ^l .
$$
For 
$\alpha = \E^{\iu \theta } \abs{\alpha }$, 
consider
\begin{flalign}  \label{in:chapo} 
{}&{}_2\lan \alpha \vert \varphi , 0 \ran = \E^{- \vert \alpha \vert ^2/2} \sum _{j=0}^{\infty } \E^{\iu j \varphi } \frac{1}{2^j j!} (- \bar{u}_1^2 + \bar{\alpha }^2)^j  \nonumber \\
{}&= \E^{- \vert \alpha \vert ^2/2} f(\varphi , 0, \bar{u}_1, \E^{-\iu \theta } \abs{\alpha }) 
= \E^{- \vert \alpha \vert ^2/2} \E ^{\E ^{\iu \varphi }(- \bar{u}_1^2 + \E^{-2\iu \theta } \abs{\alpha }^2)/2} \\
{}&= \E ^{-\E ^{\iu \varphi }\bar{u}_1^2} \E^{- \vert \alpha \vert ^2/2} \E ^{\E ^{\iu (\varphi - 2\theta )}\abs{\alpha }^2/2}
= \E ^{-\E ^{\iu \varphi }\bar{u}_1^2} \E^{- \vert \alpha \vert ^2/2} \sum _{j=0}^{\infty } \E^{\iu j (\varphi - 2\theta )} \frac{(\abs{\alpha }^2/2)^j}{j!}  \nonumber 
\end{flalign}
which follows from
$$
\vert \varphi , 0 \ran = \sum _{j=0}^{\infty } \E ^{\iu j \varphi } \vert j, j \ran 
= \sum _{j=0}^{\infty } \E^{\iu j \varphi } \frac{1}{2^j j!} (- \bar{u}_1^2 + \bar{u}_2^2)^j  
= \E ^{\E ^{\iu \varphi }(- \bar{u}_1^2 + \bar{u}_2^2)/2} .
$$
and ${}_2 \lan \alpha \vert f \ran = \E^{- \abs{\alpha }^2/2} f(\bar{\alpha })$.

If $\E ^{\iu \varphi } = \E^{2\iu \theta }$ then
$$
\E^{- \vert \alpha \vert ^2/2} \E ^{\E ^{\iu \varphi }(- \bar{u}_1^2 + \E^{-2\iu \theta } \abs{\alpha }^2)/2}
= \E ^{- \E^{2\iu \theta } \bar{u}_1^2/2} .
$$ 
It is interesting that $\E ^{- \E^{2\iu \theta } \bar{u}_1^2/2}$ is a generalized eigen-function of $\E^{-\iu \theta } c_1 + \E^{\iu \theta }c_1^*$ with eigen-value $0$.
For other $\varphi $,
$$
{}_2\lan \alpha \vert \varphi , 0 \ran 
= \E ^{-\E ^{\iu \varphi }\bar{u}_1^2} \psi _{P}(t)
= \E ^{-\E ^{\iu \varphi }\bar{u}_1^2} \E ^{(\E^{\iu t} - 1) m} \rightarrow 0, \ t = \varphi - 2\theta 
$$
as $m = \abs{\alpha }^2/2 \rightarrow \infty $, where $\psi _{P}(t)$ is the characteristic function of the Poisson distribution $P(j; m) = \E ^{-m} \frac{m^j}{j!}$.  This suggests the convergence 
\beq \label{in:charg}
  \abs{\alpha } {}_2 \lan \alpha \vert \Pi ^0 \vert \alpha \ran _2 = \frac{\abs{\alpha }}{2 \pi } \int _{2\theta - \pi}^{2\theta + \pi } {}_2\lan \alpha \vert \varphi , 0 \ran \lan \varphi , 0 \vert \alpha \ran _2 d \varphi 
  \rightarrow \frac{1}{\sqrt{2\pi }} \E ^{- \E^{2\iu \theta } \bar{u}_1^2/2} \overline{\E ^{- \E^{2\iu \theta } \bar{v}_1^2/2}}
\eeq
as $\abs{\alpha } \rightarrow \infty $, since the Poisson distribution $P(j; m) = \E ^{-m} \frac{m^j}{j!}$ approximates the Gaussian distribution $G(j; m, \sigma ^2) = \frac{1}{\sqrt{2 \pi m}} \E ^{- (j - m)^2/2 \sigma ^2}$ ($\sigma ^2 = m = \abs{\alpha }^2/2$) well and 
$$
  (\sqrt{2m}/\sqrt{2\pi }) \overline{\psi _{P}(t)} \psi _{P}(t) \approx (\sqrt{2m}/\sqrt{2\pi }) \overline{\psi _{G}(t)} \psi _{G}(t) = (\sqrt{2m}/\sqrt{2\pi }) \E ^{- m t^2} \rightarrow \delta (t)
$$ as $m \rightarrow \infty $. In  Corollary \ref{di:cordirac}, the relation $(\sqrt{2m}/\sqrt{2\pi }) \overline{\psi _{P}(t)} \psi _{P}(t) \rightarrow \delta (t)$ is proved without using the relation between $P(j; m)$ and $G(j; m, \sigma ^2)$.

In Section \ref{direct}, we will prove the convergence
\beq \label{ih:formulamain}
\abs{\alpha } {}_2 \lan \alpha \vert \Pi ^l \vert \alpha \ran _2
\rightarrow \frac{1}{\sqrt{2 \pi }} \E ^{-x^2/4} \E ^{x \E ^{\iu \theta } \bar{u}_1} \E ^{ - \E ^{2 \iu \theta } \bar{u}_1^2/2}
\overline{\E ^{-x^2/4} \E ^{x \E ^{\iu \theta } \bar{v}_1}\E ^{ - \E ^{2 \iu \theta } \bar{v}_1^2/2}} ,
\eeq
as $\abs{\alpha } \rightarrow \infty $ for fixed $x \in \R$ such that $l = [x\abs{\alpha }]$, the main theorem, Theorem \ref{di:main2} which implies (\ref{in:strongcon}).  For the proof we compare the formula
$$
{}_2\lan \alpha \vert \varphi , l \ran = \E ^{- \abs{\alpha }^2/2} \left[ 
\sum _{j=0}^{\infty } \E ^{\iu j \varphi } \frac{\sqrt{j!}}{\sqrt{2^l(l+j)!}} (\bar{u}_1 + \bar{\alpha })^l \frac{1}{2^j j!} (- \bar{u}_1^2 + \bar{\alpha }^2)^j \right] 
$$
with (\ref{in:chapo}).  The terms $\frac{\sqrt{j!}}{\sqrt{2^l(l+j)!}} (\bar{u}_1 + \bar{\alpha })^l$ will change ${}_2\lan \alpha \vert \varphi , 0 \ran $ into ${}_2\lan \alpha \vert \varphi , l \ran $.  We find a useful formula (\ref{di:psum}):
$$
\frac{(\bar{u}_1 + \bar{\alpha })^l}{(- \bar{u}_1^2 + \bar{\alpha }^2)^{l/2}} \E^{- \vert \alpha \vert ^2/2} \sum _{j=l/2}^{\infty } \E^{\iu j \varphi } 
\frac{j!}{\sqrt{(j+l/2)! (j - l/2)!}}
\frac{1}{2^j j!} (- \bar{u}_1^2 + \bar{\alpha }^2)^j = \E^{\iu l \varphi /2} {}_2\lan \alpha \vert \varphi , l \ran .
$$

$\frac{(\bar{u}_1 + \bar{\alpha })^l}{(- \bar{u}_1^2 + \bar{\alpha }^2)^{l/2}}$ gives the term
$\lim _{\abs{\alpha } \rightarrow \infty } \frac{(\bar{u}_1 + \bar{\alpha })^l}{(- \bar{u}_1^2 + \bar{\alpha }^2)^{l/2}}
= \E ^{x \E ^{\iu \theta }\bar{u}_1}$, where we use the formula $\lim _{n \rightarrow \infty }(1 + a/n)^n = \E ^a$, and $\frac{j!}{\sqrt{(j+l/2)! (j - l/2)!}}$ will give the normalization term $\E ^{-x^2/4}$.  This can be seen by Chebyshev's inequality for Poisson distribution (see Lemma \ref{di: cpsum}) and  $\lim _{m \rightarrow \infty } \frac{j!}{\sqrt{(j+l/2)! (j - l/2)!}} = \E ^{-x^2/4\mu }$ for $l = x\sqrt{2m}$ and $j = \mu m$ which follows from the Stirling's approximation $n! \sim \sqrt{2 \pi n} (n/\E )^n$ and the formula $\lim _{n \rightarrow \infty }(1 + a/n)^n = \E ^a$.
The Chebyshev's inequality shows that the summation $\sum _{0 \leq j < m - \lambda \sqrt{m}} + \sum _{m + \lambda \sqrt{m} < j}$ is small and $\sum _{j=l/2}^{\infty }$ can be replaced by 
$\sum _{j=m - \lambda \sqrt{m}}^{m + \lambda \sqrt{m}}$ for large $\lambda $.  If $j \in [m - \lambda \sqrt{m}, m + \lambda \sqrt{m}]$, then $\mu = j/m \in [1 - \lambda /\sqrt{m}, 1 + \lambda /\sqrt{m}]$.  So, $\E ^{-x^2/4\mu }$ can be replaced by $\E ^{-x^2/4}$ for large $m = \abs{\alpha }^2/2$. 
In this way, we have the main proposition Proposition \ref{di:mainprop}:
\begin{align*}
&\E^{\iu l \varphi /2} {}_2\lan \alpha \vert \varphi , l \ran \approx \E ^{x \E ^{\iu \theta }\bar{u}_1} \E ^{-x^2/4} \E^{- \vert \alpha \vert ^2/2} \sum _{j=l/2}^{\infty } \E^{\iu j \varphi } \frac{1}{2^j j!} (\bar{u}_1^2 + \E ^{-2\iu \theta }\abs{\alpha }^2)^j
\\
&\approx \E ^{x \E ^{\iu \theta }\bar{u}_1} \E ^{- \E ^{\iu \varphi } \bar{u}_1^2/2} \E ^{-x^2/4} \E^{- \vert \alpha \vert ^2/2} \sum _{j=0}^{\infty } \E^{\iu j (\varphi - 2\theta )} \frac{\abs{\alpha }^2/2}{j!} 
\end{align*}
for sufficiently large $l = x\abs{\alpha }$.
The above formula and Corollary \ref{di:cordirac} imply the formula (\ref{ih:formulamain}): Theorem \ref{di:main2}.

In Section \ref{conclusion}, It is shown that the balanced homodyne measurement causes the quantum teleportation.

\section{Holomorphic representation of CCR} \label{holomorph}

Later we will use the holomorphic representation of the canonical commutation relation defined below {(see \cite{FS80})}.
\begin{defn} \label{di:defholo}
	Consider the Hilbert space $L^2(\C ^n, d\mu _n)$ with
	$$
	u_j = x_j + \iu y_j, \ \ d \mu _n = \prod _{j=1}^n \E ^{-\bar{u}_ju_j} \frac{d\bar{u}_j du_j}{2\pi \iu }
	= \prod _{j=1}^n \E ^{-(x_j^2 + y_j^2)} \frac{d x_j d y_j}{\pi }, \ \ 
	\int _{\C ^n} d \mu _n = 1
	$$
	and its subspace $\h$ generated by holomorphic functions of $\bar{u} = (\bar{u}_1, \ldots , \bar{u}_n)$ (anti-holomorphic functions of $u$).  For $f, g \in \h $ define the inner product by
	$$
	\lan f \vert g \ran = \int _{\C ^n} \overline{f(\bar{u})} g(\bar{u}) d \mu _n = \int _{\C ^n} \overline{f(\bar{u})} g(\bar{u}) \prod _{j=1}^n \E ^{-\bar{u}_ju_j} \frac{d\bar{u}_j du_j}{2\pi \iu } .
	$$
	Then the multiplication operator $c_j^*: f(\bar{u}) \rightarrow \bar{u}_j f(\bar{u})$ and the differentiation 
	operator $c_j: f(\bar{u}) \rightarrow \partial f(\bar{u})/ \partial \bar{u}_j$ ($\partial /\partial \bar{u}_j = (1/2)(\partial /\partial x_j + \iu \partial /\partial y_j)$) are adjoint to each other and satisfy the canonical commutation relations (CCR)
	$$
	[c_j, c_k] = [c_j^*, c_k^*] = 0, \ \ [c_j, c_k^*] = \delta _{j, k} .
	$$
	This representation of the commutation relations is called the holomorphic representation.
\end{defn}
\begin{rem} \label{di:remcons}
	$\{ \bar{u}_1^{k_1} \cdots \bar{u}_n^{k_n} /\sqrt{k_1! \cdots k_n!} \}_{k_1, \ldots k_n =0}^{\infty }$ is a complete orthonormal system of $\h$.  The following calculations show the orthonormality: If $k \leq l$ then we have
\begin{align*}
	&\int _{\C ^n} \overline{\bar{u}_j^k} \bar{u}_j^l d\mu _n = \int _{\C ^n} u_j^k \bar{u}_j^l d\mu _n
	= \int _{\C ^n} u_j^k \bar{u}_j^l \prod _{j=1}^n \E ^{-\bar{u}_ju_j} \frac{d\bar{u}_j du_j}{2\pi \iu }
	\\
	&= \int _{\C ^n} u_j^k \left( - \frac{\partial }{\partial u_j} \right) ^l \prod _{j=1}^n \E ^{-\bar{u}_ju_j} \frac{d\bar{u}_j du_j}{2\pi \iu }
	= \int _{\C ^n} \left(  \frac{\partial ^l}{(\partial u_j )^l}  u_j^k \right) \prod _{j=1}^n \E ^{-\bar{u}_ju_j} \frac{d\bar{u}_j du_j}{2\pi \iu } = k! \delta _{k, l},
	\end{align*}
	if $k > l$ we calculate
	$$
	\int _{\C ^n} \bar{u}_j^l \left( - \frac{\partial }{\partial \bar{u}_j} \right) ^k \prod _{j=1}^n \E ^{-\bar{u}_ju_j} \frac{d\bar{u}_j du_j}{2\pi \iu }
	=   \int _{\C ^n} \left( \frac{\partial ^k}{(\partial \bar{u}_j)^k} \bar{u}_j^l  \right) \prod _{j=1}^n \E ^{-\bar{u}_ju_j} \frac{d\bar{u}_j du_j}{2\pi \iu } = 0.
	$$
	The completeness follows from the fact that  holomorphic functions $f(\bar{u})$ of $\bar{u}$ on $\C ^n$ have power series expansions
	$$
	f(\bar{u}) = \sum _{k_1, \ldots , k_n = 0}^{\infty } c_{k_1 \ldots k_n} \bar{u}_1^{k_1} \cdots \bar{u}_n^{k_n} .$$
\end{rem}
\begin{prop} \label{di:substitute}
	Let $\bar{\alpha }\cdot u = \sum _{j=1}^n \bar{\alpha }_j u_j$.  Then we have
	$$
	\int _{\C ^n} \E ^{\bar{\alpha }\cdot u} f(\bar{u}) d \mu _n = f(\bar{\alpha }) .
	$$
\end{prop}
\begin{proof}
	It follows from the equalities
	\begin{align*}
&	\int _{\C ^n} \E ^{\bar{\alpha }\cdot u} \bar{u}_1^{k_1} \cdots \bar{u}_n^{k_n} d \mu _n
	= \prod _{j=1}^n \int _{\C } \E ^{\bar{\alpha }_j u_j} \bar{u}_j^{k_j} \E ^{-\bar{u}_ju_j} \frac{d\bar{u}_j du_j}{2\pi \iu }
	\\
&	= \prod _{j=1}^n \int _{\C } \sum _{l_j = 0}^{\infty} \frac{(\bar{\alpha }_j u_j)^{l_j}}{l_j !} \bar{u}_j^{k_j} \E ^{-\bar{u}_ju_j} \frac{d\bar{u}_j du_j}{2\pi \iu } = \prod _{j=1}^n \bar{\alpha }_j^{k_j} = \bar{\alpha }_1^{k_1} \cdots \bar{\alpha }_n^{k_n}
	\end{align*}
	that
	\begin{align*}
	\int _{\C ^n} \E ^{\bar{\alpha }\cdot u} f(\bar{u}) d \mu _n &= \int _{\C ^n} \E ^{\bar{\alpha }\cdot u} \sum _{k_1, \ldots , k_n = 0}^{\infty } c_{k_1 \ldots k_n} \bar{u}_1^{k_1} \cdots \bar{u}_n^{k_n} d \mu _n
	\\
&= \sum _{k_1, \ldots , k_n = 0}^{\infty } c_{k_1 \ldots k_n} \bar{\alpha }^{k_1} \cdots \bar{\alpha }^{k_n} = f(\bar{\alpha }) .
\end{align*}
\end{proof}
\begin{cor} \label{di:idop}
	$\E ^{\bar{u} \cdot v}$ is the integral kernel of the identity operator.
\end{cor}
\begin{proof}
	$$
	\h \ni f(\bar{v}) \mapsto \int _{\C ^n} \E ^{\bar{u} \cdot v} f(\bar{v}) \prod _{j=1}^n \E ^{-\bar{v}_jv_j} \frac{d\bar{v}_j dv_j}{2\pi \iu } = f(\bar{u}) \in \h .
	$$
\end{proof}
\begin{rem} \label{in:costrep}
The holomorphic representation of the coherent state
$$
	\prod _{j=1}^n \E^{- \abs{\alpha _j}^2/2} \E^{\alpha _j c_j^*} \vert 0 \ran , \  \alpha _j \in \C 
	$$ is $$
	\vert \alpha \ran = \E^{- \abs{\alpha }^2/2} \E^{\alpha \cdot \bar{u}} , \ \abs{\alpha }^2 = \bar{\alpha } \cdot \alpha 
$$
and it follows from Proposition \ref{di:substitute} that
$$
	\lan \alpha \vert f \ran = \E^{- \abs{\alpha }^2/2} f(\bar{\alpha })
$$
for $\vert f \ran $ represented by the holomorphic function $f(\bar{u})$ of $\bar{u}$ in the holomorphpic representation.  If $A(\bar{u}, v)$ is the kernel function of an operator $A$ in the holomorphic representation, then $\lan \alpha \vert A \vert \beta \ran = \E^{- \abs{\alpha }^2/2} \E^{- \abs{\beta }^2/2} A(\bar{\alpha }, \beta )$.
	
	Therefore, the holomorphic representation is sometimes called the coherent-state representation.
\end{rem}
The following equalities for the coherent state follow from Proposition \ref{di:substitute} and Corollary \ref{di:idop}.
\beq
  \lan \alpha \vert \beta \ran = \int _{\C ^n} \E^{- \abs{\alpha }^2/2} \E^{\bar{\alpha } \cdot u} 
  \E^{- \abs{\beta }^2/2} \E^{\beta \cdot \bar{u}} \prod _{j=1}^n \E ^{-\bar{u}_ju_j} \frac{d\bar{u}_j du_j}{2\pi \iu } = \E^{- \abs{\alpha }^2/2} \E^{- \abs{\beta }^2/2} \E^{\bar{\alpha } \cdot \beta },
\eeq
\beq \label{di:idop1}
  \int _{\C ^n} \vert \alpha \ran \lan \alpha \vert \prod _{j=1}^n \frac{d\bar{u}_j du_j}{2\pi \iu }
  = \int _{\C ^n} \E^{\alpha \cdot \bar{u}} \E^{\bar{\alpha } \cdot v} 
  \prod _{j=1}^n \E ^{-\bar{\alpha }_j \alpha _j} \frac{d\bar{\alpha }_j d\alpha _j}{2\pi \iu }
  = \E ^{\bar{u} \cdot v} = I .
\eeq
\begin{prop}
The operator 
\beq \label{di:opxitheta}
  \xi _j(\theta ) = \E^{-\iu \theta } c_j + \E^{\iu \theta }c_j^* = \E^{-\iu \theta } \frac{\partial }{\partial \bar{u}_j}  + \E^{\iu \theta } \bar{u}_j .
\eeq
has
\beq \label{ho:eigenvec}
  \vert \theta ; r\ran _j = \E ^{- r^2 /4} \E ^{r \E ^{\iu \theta } \bar{u}_j } \E ^{ - \E ^{2 \iu \theta } \bar{u}_j^2/2}
\eeq
as an eigen-vector 
with eigen-value $r$, and these eigen-vectors satisfy
\beq 
  \int _{-\infty }^{\infty } \vert \theta ; r\ran _j {}_j \lan \theta ; r \vert  dr
  = \sqrt{2 \pi } I
\eeq
and
\beq \label{ho:cons}
   {}_j \lan \theta ; r  \vert \theta ; s\ran _j = \sqrt{2 \pi } \delta (r - s).
\eeq

\end{prop}
\begin{proof}
It is easy to see that
$$
\left( \frac{\partial }{\partial \bar{u}_j}  +  \bar{u}_j \right) \E ^{r \bar{u}_j } \E ^{ - \bar{u}_j^2/2}
= r \E ^{r \bar{u}_j } \E ^{ - \bar{u}_j^2/2} ,
$$
and
$$
\left( \E^{-\iu \theta } \frac{\partial }{\partial \bar{u}_j}  + \E^{\iu \theta } \bar{u}_j \right) 
\E ^{r \E ^{\iu \theta } \bar{u}_j } \E ^{ - \E ^{2 \iu \theta } \bar{u}_j^2/2}
= r \E ^{r \E ^{\iu \theta } \bar{u}_j } \E ^{ - \E ^{2 \iu \theta } \bar{u}_j^2/2},
$$
that is, the holomorphic representation of {$\vert 0 ; r\ran _j$ is $C \E ^{r \bar{u}_j } \E ^{ - \bar{u}_j^2/2}$} and $\vert \theta ; r\ran _j$ is $C \E ^{r \E ^{\iu \theta } \bar{u}_j } \E ^{ - \E ^{2 \iu \theta } \bar{u}_j^2/2}$ for some constant $C$.

For $\vert \pi/2 ; r\ran _j$ = $\E ^{- r^2 /4} \E ^{\iu r \bar{u}_j } \E ^{\bar{u}_j^2/2}$, we have
	
\begin{align*}
 & \int _{-\infty }^{\infty } \vert \pi /2 ; r\ran _j {}_j \lan \pi /2 ; r \vert  dr
  = \int _{-\infty }^{\infty }  \E ^{- r^2 /4} \E ^{\iu r \bar{u}_j } \E ^{\bar{u}_j^2/2} \overline{\E ^{- r^2 /4} \E ^{\iu r \bar{v}_j } \E ^{\bar{v}_j^2/2}} dr
\\
 & = \int _{-\infty }^{\infty } \E ^{- r^2 /2} \E ^{\iu r(\bar{u}_j - v_j)} \E ^{\bar{u}_j^2/2} \E ^{v_j^2/2} dr
  = \sqrt{2 \pi } \E ^{- (\bar{u}_j - v_j)^2/2} \E ^{\bar{u}_j^2/2} \E ^{v_j^2/2} = \sqrt{2 \pi } \E ^{\bar{u}_j v_j} ,
\end{align*}
and
\begin{align*}
  &{}_j \lan \pi /2 ; r  \vert \pi /2 ; s\ran _j = \int _{\C }  \overline{\E ^{- r^2 /4} \E ^{\iu r \bar{u}_j } \E ^{\bar{u}_j^2/2}} \E ^{- s^2 /4} \E ^{\iu s \bar{u}_j } \E ^{\bar{u}_j^2/2} \E ^{-\bar{u}_ju_j} \frac{d\bar{u}_j du_j}{2\pi \iu }
\\
&= \E ^{- r^2 /4} \E ^{- s^2 /4} \int _{\C }  \E ^{-\iu r u_j } \E ^{u_j^2/2} \E ^{\iu s \bar{u}_j } \E ^{\bar{u}_j^2/2} \E ^{-\bar{u}_ju_j} \frac{d\bar{u}_j du_j}{2\pi \iu }
\\
&= \E ^{- (r^2 + s^2)/4}  \int _{\R ^2}  \E ^{-\iu (r - s) x_j } \E ^{(r + s) y_j } \E ^{-2 y_j^2} \frac{d x_j d y_j}{\pi }
\\
&= \E ^{- (r - s)^2 /8} \frac{1}{\sqrt{2 \pi }} \int _{-\infty }^{\infty }  \E ^{-\iu (r - s) x_j } d x_j 
= \E ^{- (r - s)^2 /8} \sqrt{2 \pi } \delta (r - s) = \sqrt{2 \pi } \delta (r - s).
\end{align*}
By the change of variables $\bar{u}_j \rightarrow \E^{\iu (\theta - \pi /2)} \bar{u}_j$, $\bar{v}_j \rightarrow \E^{\iu (\theta - \pi /2)} \bar{v}_j$, we get
\begin{align*}
 & \int _{-\infty }^{\infty } \vert \theta ; r\ran _j {}_j \lan \theta ; r \vert  dr =
\\
 &  \int _{-\infty }^{\infty }  \E ^{- r^2 /4} \E ^{\iu r \E^{\iu (\theta - \pi /2)} \bar{u}_j } \E ^{\E^{2\iu (\theta - \pi /2)} \bar{u}_j^2/2} \overline{\E ^{- r^2 /4} \E ^{\iu r \E^{\iu (\theta - \pi /2)} \bar{v}_j } \E ^{\E^{2\iu (\theta - \pi /2)}\bar{v}_j^2/2}} dr
\\
 & = \sqrt{2 \pi } \E ^{\E^{\iu (\theta - \pi /2)}\bar{u}_j \E^{-\iu (\theta - \pi /2)}v_j} = \sqrt{2 \pi } \E ^{\bar{u}_j v_j} = \sqrt{2 \pi }I,
\end{align*}
and
\begin{align*}
  & {}_j \lan \theta ; r  \vert \theta ; s\ran _j
    = \\
    & \int _{\C }  \overline{\E ^{- r^2 /4} \E ^{\iu r \E^{\iu (\theta - \pi /2)} \bar{u}_j } \E ^{\E^{2\iu (\theta - \pi /2)} \bar{u}_j^2/2}} \E ^{- s^2 /4} \E ^{\iu s \E^{\iu (\theta - \pi /2)} \bar{u}_j } \E ^{\E^{2\iu (\theta - \pi /2)} \bar{u}_j^2/2} \E ^{-\bar{u}_ju_j} \frac{d\bar{u}_j du_j}{2\pi \iu }
\\ 
{} & = \sqrt{2 \pi } \delta (r - s),
\end{align*}
where used the fact that $\E ^{-\bar{u}_ju_j}d\bar{u}_j du_j$ is invariant under the above change of variables.
This completes the proof.
\end{proof}
The above proposition shows the spectral resolution of $\xi _j(\theta )$:
$$
  \xi _j(\theta ) = \frac{1}{\sqrt{2 \pi }} \int _{-\infty }^{\infty } r \E ^{- r^2 /4} \E ^{r \E ^{\iu \theta } \bar{u}_j } \E ^{ - \E ^{2 \iu \theta } \bar{u}_j^2/2}
  \overline{\E ^{- r^2 /4} \E ^{r \E ^{\iu \theta } \bar{v}_j } \E ^{ - \E ^{2 \iu \theta } \bar{v}_j^2/2}} dr .
$$
\begin{rem} \label{ho:hbs}
	Since $x_0 = (c_0 + c_0^*)/\sqrt{2} = \xi _0(0)/\sqrt{2}$ and $p_1 = (c_1 - c_1^*)/\sqrt{2} \iu = \xi _1(\pi /2)/\sqrt{2}$, we have $\vert 0 ; \sqrt{2} x_- \ran _0 = \vert x_- \ran _0$ and $\vert \pi /2 ; \sqrt{2} p_+ \ran _1 = \vert p_+ \ran _1$ (see (\ref{in:01hbs})).
\end{rem}
Let
$$
  \R \ni \theta \mapsto \exp{\theta \begin{pmatrix} 0 & 1 \\ -1 & 0 \end{pmatrix}}
  = \begin{pmatrix} \cos \theta & \sin \theta \\ -\sin \theta  & \cos \theta \end{pmatrix} \in U(2)
$$
be a one-parameter subgroup in $U(2)$, the group of $2 \times 2$ unitary matrices,  and 
its action on $\h _0 \hat{\otimes } \h _1$ is
\beq \label{in:utheta}
  U(\theta ): \h _0 \hat{\otimes } \h _1 \ni f(\bar{u}_0, \bar{u}_1) \mapsto f(\bar{u}_0 \cos \theta + \bar{u}_1 \sin \theta , -\bar{u}_0 \sin \theta + \bar{u}_1 \cos \theta ) \in \h _0 \hat{\otimes } \h _1 .
\eeq
Its infinitesimal generator is
$$
  \frac{d}{d \theta } U(\theta ) f _{\vert \theta = 0} = \left( - \bar{u}_0 \frac{\partial }{\partial \bar{u}_1}
  + \bar{u}_1 \frac{\partial }{\partial \bar{u}_0}\right) f = (- c_0^* c_1 + c_1^* c_0)f ,
$$
and therefore $U(\theta )$ = $\exp \theta (- c_0^* c_1 + c_1^* c_0)$ and $U(\pi /4)$ = $\exp (\pi /4) (- c_0^* c_1 + c_1^* c_0)$ = $\E ^{\iu H_{hbs}^{01}}$ (see (\ref{in:01hbs})).
In the same reasoning we get
(see (\ref{gbbasis}))
$$
  \h _0 \ni f(\bar{u}_0) \mapsto  \E ^{-\abs{\alpha }^2/2} \E ^{\alpha \bar{u}_0} f(\bar{u}_0 - \bar{\alpha }) = \E ^{\alpha c_0^* - \bar{\alpha }c_0} f(\bar{u}_0) = D_0(\alpha )f(\bar{u}_0) ,
$$
and
$$
  \sum _{n=0}^{\infty } (D_0(\alpha ) \vert n \ran _0) \otimes \vert n \ran _1 = \sum _{n=0}^{\infty } \E ^{-\abs{\alpha }^2/2} \E ^{\alpha \bar{u}_0} \frac{(\bar{u}_0 - \bar{\alpha })^n \bar{u}_1^n}{n!} = \E ^{-\abs{\alpha }^2/2} \E ^{\alpha \bar{u}_0} \E^{\bar{u}_0 \bar{u}_1} \E^{ - \bar{\alpha } \bar{u}_1}.
$$
Applying $U(\pi /4) = \E ^{\iu H_{hbs}^{01}}$ gives
\begin{align*}
 & \E ^{\iu H_{hbs}^{01}} \sum _{n=0}^{\infty } (D_0(\alpha ) \vert n \ran _0) \otimes \vert n \ran _1 
 = \E ^{-\abs{\alpha }^2/2} \E ^{\alpha (\bar{u}_0 + \bar{u}_1)/\sqrt{2}} \E^{(\bar{u}_0 + \bar{u}_1) (-\bar{u}_0 + \bar{u}_1)/2} \E^{ - \bar{\alpha } (-\bar{u}_0 + \bar{u}_1)/\sqrt{2}}
\\
 & = \E ^{-\abs{\alpha }^2/2} \E ^{(\alpha + \bar{\alpha }) \bar{u}_0 /\sqrt{2}} \E^{(-\bar{u}_0^2 + \bar{u}_1^2)/2} \E^{(\alpha - \bar{\alpha }) \bar{u}_1 /\sqrt{2}}
  =  \pi ^{1/2} \vert x_- \ran _0 \otimes \vert p_+ \ran _1 , 
\end{align*}
where $\alpha = x_- + \iu p_+$.  This proves (\ref{in:hbsxp}).
\section{Balanced homodyne measurement} \label{balance}

Let
\beq \label{di:opXi}
\Xi  = c_2^*c_1 + c_1^*c_2 = \bar{u}_2 \frac{\partial }{\partial \bar{u}_1} + \bar{u}_1 \frac{\partial }{\partial \bar{u}_2} 
\eeq
be a {symmetric} operator on $\h = \h _1 \hat{\otimes } \h _2$ generated by $\bar{u} = (\bar{u}_1, \bar{u}_2)$
where the space $\h _j$ is generated by $\bar{u}_j$ ($j = 1, 2$). 
Then the function
$$
(\bar{u}_1 + \bar{u}_2)^m(- \bar{u}_1 + \bar{u}_2)^n
$$
is an eigen function of $\Xi $ with eigen value $m - n$.  In addition, the set of functions
$$
\vert m, n \ran  = \frac{1}{\sqrt{2^m 2^n m! n!}} (\bar{u}_1 + \bar{u}_2)^m(- \bar{u}_1 + \bar{u}_2)^n, \ \ 0 \leq m, n \in {\mathbb Z} 
$$
constitutes an orthonormal basis of the Fock space $\h $ of $c_1$ and $c_2$, which follows from  Remark \ref{di:remcons} since the above functions are expressed as
$$
\frac{1}{\sqrt{m! n!}} \bar{v}_1^m \bar{v}_2^n, \ \ 0 \leq m, n \in {\mathbb Z} 
$$ 
by the change of variables 
\beq \label{di:chvu}
\bar{v}_1 = (\bar{u}_1 + \bar{u}_2)/\sqrt{2}, \ \ \bar{v}_2 = (- \bar{u}_1 + \bar{u}_2)/\sqrt{2}
\eeq
and the measure $\mu _n$ ($n = 2$) is invariant under the change of variables (\ref{di:chvu}).

The eigen space of the operator $\Xi $ associated with the eigen-value $l \in {\mathbb Z}$ is the space spanned by the vectors
$$
\vert l + j, j \ran , \ 0 \leq j \in {\mathbb Z}
$$
for $l \geq 0$ and
$$
\vert j, j - l \ran , \ 0 \leq j \in {\mathbb Z}
$$
for $l \leq 0$.

The orthogonal projection $\Pi ^l$ onto the eigen-space with eigen-value $l$ is defined by
$$
\Pi ^l = \sum _{j = 0}^{\infty } \vert l + j, j \ran \lan l + j, j \vert 
$$
for $l \geq 0$, and
$$
\Pi ^l = \sum _{j = 0}^{\infty } \vert j, j - l \ran \lan j, j - l \vert 
$$
for $l \leq 0$.

The operator $\Xi  = c_2^*c_1 + c_1^*c_2$ has the spectral resolution
$$
\Xi  = \sum _{l = - \infty }^{\infty } l \Pi ^l,
$$
and {therefore $\Xi $ is a self-adjoint operator.  The} operator $X_{\alpha } = \Xi /\abs{\alpha }$ has the spectral resolution
$$
X_{\alpha } = \Xi /\abs{\alpha } = \sum _{l = - \infty }^{\infty } (l/\abs{\alpha }) \Pi ^l .
$$

Now we have expressed the projections $\Pi ^l$ and the generalized vectors $\vert \theta ; r\ran _1$ in the holomorphic representation.  But in order to prove 
(see (\ref{in:strongcon}))
$$
\abs{\alpha } \Pi ^l (\vert \phi \ran _1 \otimes \vert \alpha \ran _2) \rightarrow  \frac{1}{\sqrt{2 \pi }} \vert \theta ; x \ran _1 
{}_1 \lan \theta ; x \vert \phi \ran _1 \otimes \vert \alpha \ran _2 
$$
for a coherent state $\vert \alpha \ran _2 \in \h _2$ ($\alpha \in \C$) and $\vert \phi \ran _1 \in \h _1$,
we need some preparation.  We begin with the useful expression of the projection operator $\Pi ^l$.

\begin{prop} \label{di:propcosl}
Introduce
\begin{flalign} \label{di:phi1}
	&\vert \varphi , l \ran  = \sum _{j=0}^{\infty } \E ^{\iu j \varphi } \vert l+j, j \ran 
	= f(\varphi , l, \bar{u}_1, \bar{u}_2)\nonumber\\
&	= \sum _{j=0}^{\infty } \E ^{\iu j \varphi } \frac{1}{\sqrt{2^{l+j} 2^j (l+j)! j!}} (\bar{u}_1 + \bar{u}_2)^{l+j}(- \bar{u}_1 + \bar{u}_2)^j \\
&	= \frac{1}{\sqrt{2^l}} (\bar{u}_1 + \bar{u}_2)^l \sum _{j=0}^{\infty } \E ^{\iu j \varphi } \frac{\sqrt{j!}}{\sqrt{(l+j)!}} \frac{1}{2^j j!} (- \bar{u}_1^2 + \bar{u}_2^2)^j
\nonumber	
\end{flalign}
	for $l \geq 0$ and
\beq \label{di:phi2}
\begin{aligned}
	\vert \varphi , l \ran & = \sum _{j=0}^{\infty } \E ^{\iu j \varphi } \vert j, j-l \ran 
	= f(\varphi , l, \bar{u}_1, \bar{u}_2) \\ 
{}&	= \frac{1}{\sqrt{2^{\abs{l}}}} (-\bar{u}_1 + \bar{u}_2)^{\abs{l}} \sum _{j=0}^{\infty } \E ^{\iu j \varphi } \frac{\sqrt{j!}}{\sqrt{(\abs{l}+j)!}} \frac{1}{2^j j!} (- \bar{u}_1^2 + \bar{u}_2^2)^j 
\end{aligned}
\eeq
	for $l \leq 0 $.  Then the series 
	 (\ref{di:phi1}) and (\ref{di:phi2}) 
	 converge to entire functions {of $\bar{u}_1$ and $\bar{u}_1$}, and we have
	$$
	\frac{1}{2 \pi } \int _{a - \pi}^{a + \pi } \vert \varphi , l \ran \lan \varphi , l \vert d \varphi 
	= \Pi ^l 
	$$
	for any $a \in \R$.
\end{prop}
\begin{proof}
	We present the proof explicitly
	for the case of $l \geq 0$, 
	for $l \leq 0$, the proof is
	the same. The inequality
	$$
	\sum _{j=0}^n \abs{ \E ^{\iu j \varphi } \vert l+j, j \ran } \leq 
	\frac{1}{\sqrt{2^l}} \abs{ \bar{u}_1 + \bar{u}_2}^l \sum _{j=0}^n \frac{1}{2^j j!} \abs{ - \bar{u}_1^2 + \bar{u}_2^2}^j
	\leq \frac{1}{\sqrt{2^l}} \abs{ \bar{u}_1 + \bar{u}_2}^l \E ^{\abs{ - \bar{u}_1^2 + \bar{u}_2^2}/2}
	$$	
	implies the uniform convergence of (\ref{di:phi1}) with respect to $\varphi $ {in $\R $ and $\bar{u}_1, \bar{u}_2$ in any compact subset of $\C $}.  Integration term by term implies
	\begin{eqnarray*}
	\int _{a - \pi }^{a + \pi } \vert \varphi , l \ran \lan \varphi , l \vert d \varphi 
&=& \int _{a - \pi }^{a + \pi } \left[ \sum _{i=0}^{\infty } \E ^{\iu i \varphi } \vert l+i, i \ran \right]
	\left[ \sum _{j = 0}^{\infty } \E ^{-\iu j \varphi } \lan l+j, j \vert \right] d \varphi
	\\
	&=& 2 \pi \sum _{j = 0}^{\infty } \vert l + j, j \ran \lan l + j, j \vert = 2\pi \Pi ^l ,
	\end{eqnarray*} \\
	where we used the relations
	$$
	\int _{a - \pi }^{a + \pi } \E ^{\iu i \varphi } \E ^{-\iu j \varphi } d \varphi = 2\pi \delta _{i, j}
	$$
	for $i, j = 0, 1, 2, \cdots $.
\end{proof}

\begin{prop} 
$$
	\vert \varphi , 0 \ran = f(\varphi , 0, \bar{u}_1, \bar{u}_2) = \E ^{\iu \varphi (- \bar{u}_1^2 + \bar{u}_2^2)/2} .
$$
\end{prop}
\begin{proof}
\begin{eqnarray*}
	\vert \varphi , 0 \ran &=& \sum _{j=0}^{\infty } \E ^{\iu j \varphi } \frac{1}{\sqrt{2^j 2^j j! j!}} (\bar{u}_1 + \bar{u}_2)^j (- \bar{u}_1 + \bar{u}_2)^j  \\
	&=& \sum _{j=0}^{\infty } \E^{\iu j \varphi } \frac{1}{2^j j!} (- \bar{u}_1^2 + \bar{u}_2^2)^j  
	= \E ^{\E ^{\iu \varphi }(- \bar{u}_1^2 + \bar{u}_2^2)/2} .
\end{eqnarray*}
\end{proof}
It follows from Proposition \ref{di:substitute} that
$$
{}_2\lan \alpha \vert \varphi , 0 \ran  = \E^{-\vert \alpha \vert ^2/2} \int \E^{\bar{\alpha }u_2}  f(\varphi , 0, \bar{u}_1, \bar{u}_2) 
\E^{- \bar{u}_2u_2} \frac{d\bar{u}_2 du_2}{2 \pi \iu } 
= \E^{-\vert \alpha \vert ^2/2} f(\varphi , 0, \bar{u}_1, \bar{\alpha }) .
$$

With $\alpha = \E^{\iu \theta } \abs{\alpha }$ we get
\beq \label{di:alcos0}
{}_2\lan \alpha \vert \varphi , 0 \ran
= \E^{- \vert \alpha \vert ^2/2} f(\varphi , 0, \bar{u}_1, \E^{-\iu \theta } \abs{\alpha }) 
= \E^{- \vert \alpha \vert ^2/2} \E ^{\E ^{\iu \varphi }(- \bar{u}_1^2 + \E^{-2\iu \theta } \abs{\alpha }^2)/2} .
\eeq

If $\E ^{\iu \varphi } = \E^{2\iu \theta }$ then
$$
\E^{- \vert \alpha \vert ^2/2} \E ^{\E ^{\iu \varphi }(- \bar{u}_1^2 + \E^{-2\iu \theta } \abs{\alpha }^2)/2}
= \E ^{- \E^{2\iu \theta } \bar{u}_1^2/2} .
$$ 
It is interesting that $\E ^{- \E^{2\iu \theta } \bar{u}_1^2/2}$ is a generalized-eigen function of $\xi _1(\theta ) = \E^{-\iu \theta } c_1 + \E^{\iu \theta }c_1^*$ with eigen-value $0$.
For other $\varphi $,
$$
{}_2\lan \alpha \vert \varphi , 0 \ran 
= \E ^{-\E ^{\iu \varphi }\bar{u}_1^2/2} \E^{- \vert \alpha \vert ^2/2} 
\E ^{\E ^{\iu \varphi }\E^{-2\iu \theta } \abs{\alpha }^2/2} \rightarrow 0
$$
as $\abs{\alpha } \rightarrow \infty $.
This suggests the convergence of 
$$
  \abs{\alpha } {}_2 \lan \alpha \vert \Pi ^0 \vert \alpha \ran _2 = \frac{\abs{\alpha }}{2 \pi } \int _{2\theta - \pi}^{2\theta + \pi } {}_2\lan \alpha \vert \varphi , 0 \ran \lan \varphi , 0 \vert \alpha \ran _2 d \varphi 
$$
to $\displaystyle \frac{1}{\sqrt{2\pi }} \E ^{- \E^{2\iu \theta } \bar{u}_1^2/2} \overline{\E ^{- \E^{2\iu \theta } \bar{v}_1^2/2}}$ as $\abs{\alpha } \rightarrow \infty $ (see (\ref{in:charg})).  

\section{Proof of the convergence of $\abs{\alpha } {}_2 \lan \alpha \vert \Pi ^{[x \abs{\alpha }]} \vert \alpha \ran _2$} \label{direct}
Our main aim is to prove the formula (\ref{ih:formulamain}) in Section \ref{intro}:  For $\alpha = \E ^{\iu \theta } \abs{\alpha }$
\beq \label{di:formulamain}
\begin{aligned}
&{}&\abs{\alpha } {}_2 \lan \alpha \vert \Pi ^{[x \abs{\alpha }]} \vert \alpha \ran _2
 \longrightarrow  \frac{1}{\sqrt{2 \pi }} \vert \theta ; x \ran _1 {}_1 \lan \theta ; x \vert \\
& = & \frac{1}{\sqrt{2 \pi }} \E ^{x \E ^{\iu \theta } \bar{u}_1} \E ^{ - \E ^{2 \iu \theta } \bar{u}_1^2/2} \E ^{-x^2/4} 
\overline{\E ^{x \E ^{\iu \theta } \bar{v}_1}\E ^{ - \E ^{2 \iu \theta } \bar{v}_1^2/2} \E ^{-x^2/4}} 
\end{aligned}
\eeq
as $\abs{\alpha } \rightarrow \infty $.  Note that the above formula immediately gives the convergence of the probability distribution (see (\ref{in:braun}))
$$
\abs{\alpha } ({}_1 \lan \beta \vert \otimes {}_2 \lan \alpha \vert ) \Pi ^{[x \abs{\alpha }]} (\vert \beta \ran _1 \otimes \vert \alpha \ran _2) \rightarrow  \frac{1}{\sqrt{2 \pi }} {}_1 \lan \beta \vert \theta ; x \ran _1 
{}_1 \lan \theta ; x \vert \beta \ran _1 = \frac{1}{\sqrt{2 \pi }} \E ^{-\{ x - (\bar{\alpha } \beta  + \alpha \bar{\beta })/\abs{\alpha }\} ^2/2} .
$$
Define 
projections $\Pi ^{(a, b]\abs{\alpha }}$, $P_R^{(a, b]}$ and $P^{(a, b]}$ by
$$
\Pi ^{(a, b]\abs{\alpha }} = \sum _{a \abs{\alpha } < l \leq b \abs{\alpha }} \Pi ^l , \ \ 
P_R^{(a,b]} = \frac{1}{\sqrt{2 \pi }} \sum _{a \abs{\alpha } < l \leq b \abs{\alpha }} \vert \theta ; l/\abs{\alpha } \ran _1
{}_1 \lan \theta ; l/\abs{\alpha } \vert  \frac{1}{\abs{\alpha }}
$$
and
$$  
P^{(a, b]} = \frac{1}{\sqrt{2 \pi }} \int _{a}^{b} \vert \theta ; r \ran _1
{}_1 \lan \theta ; r \vert  dr
$$
respectively.

Before we prove 
(\ref{di:formulamain}), we show that 
formula (\ref{di:formulamain}) implies the following proposition:
\begin{prop} \label{di:pitop}
	Let $\vert \beta \ran _1$ be a coherent state, and $\vert \psi \ran = \vert \beta \ran _1 \otimes \vert \alpha \ran _2$.
	Then	
	$$
	\norm{ \Pi ^{(a, b]\abs{\alpha }} \vert \psi \ran - (P^{(a, b]} \otimes I) \vert \psi \ran } ^2 \rightarrow 0
	$$
	as $\abs{\alpha } \rightarrow \infty $.
\end{prop}
\begin{proof}
In order to prove the proposition, we will show that each term of the right hand side of the following formula converges to ${}_1 \lan \beta \vert P^{(a, b]} \vert \beta \ran _1$ as $\abs{\alpha } \rightarrow \infty $.
	\begin{flalign*}
{}&	\norm{ \Pi ^{(a, b]\abs{\alpha }} \vert \psi \ran - P^{(a, b]}  (\vert \beta \ran _1 \otimes \vert \alpha \ran _2) } ^2
\\
{}&= \lan \psi \vert \Pi ^{(a, b]\abs{\alpha }} \vert \psi \ran 
	+  {}_1 \lan \beta \vert P^{(a, b]} \vert \beta \ran _1  {}_2 \lan \alpha \vert \alpha \ran _2 
	\\
{}&	- \lan \psi \vert \Pi ^{(a, b]\abs{\alpha }} (P^{(a, b]} \vert \beta \ran _1 \otimes \vert \alpha \ran _2)
	- ({}_1 \lan \beta \vert P^{(a, b]}
	\otimes {}_2 \lan \alpha \vert )\Pi ^{(a, b]\abs{\alpha }} \vert \psi \ran .
	\end{flalign*}
	Since the formula (\ref{di:formulamain}) implies
$$
   {}_1 \lan \beta \vert \abs{\alpha } {}_2 \lan \alpha \vert \Pi ^{[x \abs{\alpha }]} \vert \alpha \ran _2 \vert \beta \ran _1
   \rightarrow  \frac{1}{\sqrt{2 \pi }} {}_1 \lan \beta \vert \theta ; x \ran _1 
   {}_1 \lan \theta ; x \vert \beta \ran _1
$$
as $\abs{\alpha } \rightarrow \infty $, we have
	\begin{flalign*}
{}&	\lan \psi \vert \Pi ^{(a, b]\abs{\alpha }} \vert \psi \ran  = \sum _{a \abs{\alpha } < l \leq b \abs{\alpha }} {}_1 \lan \beta \vert {}_2 \lan \alpha \vert \Pi ^l \vert \alpha \ran _2 \vert \beta \ran _1 \nonumber \\
{}&	 \rightarrow \frac{1}{\sqrt{2 \pi }} \sum _{a \abs{\alpha } < l \leq b \abs{\alpha }} {}_1 \lan \beta \vert \theta ; l/\abs{\alpha } \ran _1
	{}_1 \lan \theta ; l/\abs{\alpha } \vert \beta \ran _1 \frac{1}{\abs{\alpha }} = {}_1 \lan \beta \vert P_R^{(a, b]} \vert \beta \ran _1 \rightarrow {}_1 \lan \beta \vert P^{(a, b]} \vert \beta \ran _1
	\nonumber
	\end{flalign*}
	as $\abs{\alpha } \rightarrow \infty $.  Next observe
$$
	\lan \psi \vert \Pi ^{(a, b]\abs{\alpha }} (P^{(a, b]} \vert \beta \ran _1 \otimes \vert \alpha \ran _2)
	= {}_1 \lan \beta \vert ({}_2 \lan \alpha \vert \Pi ^{(a, b]\abs{\alpha }} \vert \alpha \ran _2 ) P^{(a, b]} \vert \beta \ran _1 
$$
$$
	\rightarrow {}_1 \lan \beta \vert P_R^{(a, b]} P^{(a, b]} \vert \beta \ran _1 = {}_1 \lan \beta \vert P_R^{(a, b]} \vert \beta \ran _1 \rightarrow {}_1 \lan \beta \vert P^{(a, b]} \vert \beta \ran _1
$$
	where we used the relation (see (\ref{ho:cons}))
\begin{flalign*}
{}&	\vert \theta ; l/\abs{\alpha } \ran _1  {}_1 \lan \theta ; l/\abs{\alpha } \vert  P^{(a, b]}
	= \frac{1}{\sqrt{2 \pi }} \int _{a}^{b} \vert \theta ; l/\abs{\alpha } \ran _1  {}_1 \lan \theta ; l/\abs{\alpha } \vert \theta ; r \ran _1
	{}_1 \lan \theta ; r \vert  dr
	\\
{}&	= \int _{a}^{b} \vert \theta ; l/\abs{\alpha } \ran _1  {}_1 \delta (l/\abs{\alpha } - r)
	{}_1 \lan \theta ; r \vert  dr = \vert \theta ; l/\abs{\alpha } \ran _1  {}_1 \lan \theta ; l/\abs{\alpha } \vert .
\end{flalign*}
	Similarly we find
\begin{flalign*}
{}&	({}_1 \lan \beta \vert P^{(a, b]}
	\otimes {}_2 \lan \alpha \vert )\Pi ^{(a, b]\abs{\alpha }} \vert \psi \ran
	= {}_1 \lan \beta \vert P^{(a, b]} ({}_2 \lan \alpha \vert \Pi ^{(a, b]\abs{\alpha }} \vert \alpha \ran _2 ) \vert \beta \ran _1 
	\\
{}&	\rightarrow {}_1 \lan \beta \vert P^{(a, b]} P_R^{(a, b]} \vert \beta \ran _1 
	= {}_1 \lan \beta \vert P_R^{(a, b]} \vert \beta \ran _1 \rightarrow {}_1 \lan \beta \vert P^{(a, b]} \vert \beta \ran _1 .
\end{flalign*}
	This proves the proposition.
\end{proof}
\begin{cor} 
	Let $\vert \phi \ran _1 \in \h _1$, and $\vert \psi \ran = \vert \phi \ran _1 \otimes \vert \alpha \ran _2$.
	Then	
	$$
	\norm{ \Pi ^{(a, b]\abs{\alpha }} \vert \psi \ran - (P^{(a, b]} \otimes I) \vert \psi \ran } ^2 \rightarrow 0
	$$
	as $\abs{\alpha } \rightarrow \infty $.
\end{cor}
\begin{proof}
	Since $\Pi ^{(a, b]\abs{\alpha }}$ and $P^{(a, b]}$ are projections (and consequently, bounded operators) and linear combinations of coherent states constitute a dense subset of $\h _1$ (see (\ref{di:idop1})), the corollary follows from the previous proposition.
\end{proof}
\begin{rem}
We can modify the postulates of projective measurement of Definition \ref{in:projmesur}  as follows (see \cite{BP02}):
	{\ } \newline \indent
Postulate 1) When one measures the observable $R = \int _{-\infty }^{\infty} r dE_r$ on  
the state $\vert \psi \ran $, the probability for obtaining the outcome $r$ in the interval $(a, b]$ is
$\lan \psi \vert E(a, b] \vert \psi \ran /\lan \psi \vert \psi \ran $,
where $E(a, b] = \int _a^b dE_r$.

Postulate 2) After the measurement the initial state $\vert \psi \ran $ changes to $E(a, b] \vert \psi \ran $.	
\end{rem}

Note that for the observable $\Xi  = \sum _{l = - \infty }^{\infty } l \Pi ^l$, $E_r = \Pi ^{(-\infty , r]}$ and $E(a, b] = \Pi ^{(a, b]}$, for the observable 
$\xi (\theta ) = \frac{1}{\sqrt{2 \pi }} \int _{-\infty }^{\infty } r \vert \theta ; r \ran _1
{}_1 \lan \theta ; r \vert  dr$, $E_r = P^{(-\infty , r]}$ and $E(a, b] = P^{(a, b]}$.

\begin{rem} \label{di:bhomodynecoll}
Let $a = x - \epsilon $ and $b = x + \epsilon $ for a small $\epsilon > 0$, and $\vert \phi \ran _1 \in \h _1$. Then the relation
$$
\Pi ^{(a, b]\abs{\alpha }} \vert \phi \ran _1 \otimes \vert \alpha \ran _2 \approx (P^{(a, b]}\vert \phi \ran _1) \otimes \vert \alpha \ran _2 
$$
shows the following:  If the balanced homodyne detector detects the difference $N_2 - N_1$ of the photon numbers $N_1$ and
$N_2$ on
the state $e^{\iu H^{12}_{\rm hbs}} \vert \phi \ran _1 \otimes \vert \alpha \ran _2$, composed of the signal $\vert \phi \ran _1$, the strong laser beam $\vert \alpha \ran _2$ and the operation of beam-splitter $e^{\iu H^{12}_{\rm hbs}}$, and the result is $l = x \abs{\alpha }$ within an error $\epsilon \abs{\alpha }$, then the state $e^{\iu H^{12}_{\rm hbs}} \vert \phi \ran _1 \otimes \vert \alpha \ran _2$ changes to $e^{\iu H^{12}_{\rm hbs}} P^{(a, b]}
 \vert \phi \ran _1 \otimes \vert \alpha \ran _2$. 
 
If ${}_1 \lan \theta ; x \vert \phi \ran _1$ is meaningful (e.g., $\vert \phi \ran _1$ is a linear combination of coherent states), $(P^{(a, b]}\vert \phi \ran _1) \otimes \vert \alpha \ran _2 \approx 2\epsilon \vert \theta ; x \ran _1
{}_1 \lan \theta ; x \vert \phi \ran _1 \otimes \vert \alpha \ran _2$
and the state $e^{\iu H^{12}_{\rm hbs}} \vert \phi \ran _1 \otimes \vert \alpha \ran _2$ changes to $2\epsilon e^{\iu H^{12}_{\rm hbs}} \vert \theta ; x \ran _1
{}_1 \lan \theta ; x \vert \phi \ran _1 \otimes \vert \alpha \ran _2$,
which is equivalent to $e^{\iu H^{12}_{\rm hbs}} \vert \theta ; x \ran _1
{}_1 \lan \theta ; x \vert \phi \ran _1 \otimes \vert \alpha \ran _2$ by the normalization (of generalized state).

For $\epsilon = 1/2 \abs{\alpha }$ the relations $\Pi ^{(a, b]\abs{\alpha }}\vert \phi \ran _1 \otimes \vert \alpha \ran _2 = \Pi ^{[x\abs{\alpha } + 1/2]}\vert \phi \ran _1 \otimes \vert \alpha \ran _2$ and $(P^{(a, b]}\vert \phi \ran _1) \otimes \vert \alpha \ran _2 \approx (1/\abs{\alpha }) \vert \theta ; x \ran _1
	{}_1 \lan \theta ; x \vert \phi \ran _1 \otimes \vert \alpha \ran _2$ show (\ref{in:strongcon}).
\end{rem}
First, we attend to 
the case of $x = 0$, i.e., $\lim _{\abs{\alpha } \rightarrow \infty }  \abs{\alpha } {}_2 \lan \alpha \vert \Pi ^0 \vert \alpha \ran _2$.  We begin with the following proposition.

\begin{prop} \label{di:propdirac}
	For a continuous function ${g}(\varphi )$ we have	
	\beq \label{di:diracd}
	\frac{\abs{\alpha }}{\sqrt{2 \pi }} \int _{- \pi }^{\pi } g(\varphi ) \E ^{(\cos \varphi - 1) \abs{\alpha }^2} d \varphi 
	\rightarrow g(0)
	\eeq
	as $\abs{\alpha } \rightarrow \infty $.	
\end{prop}
\begin{proof}
Taylor's theorem tells us that the relation
$$
	\cos \varphi - 1 = - \frac{1}{2} \varphi ^2 + \frac{\cos \theta \varphi }{4!} \varphi ^4 
$$
holds for some $0 < \theta < 1$, and consequently, for $-\pi /2 \leq - \delta \leq \varphi \leq \delta \leq \pi /2$,
$$
	- \frac{1}{2} \varphi ^2 \leq - \frac{1}{2} \varphi ^2 + \frac{\cos \varphi }{4!} \varphi ^4 \leq \cos \varphi - 1 \leq - \frac{1}{2} \varphi ^2 + \frac{1}{4!} \varphi ^4 \leq - \frac{1}{2} \varphi ^2 + \frac{1}{4!} \delta ^4 
$$
and 
$$
	\E ^{-\varphi ^2 \abs{\alpha }^2/2} \leq \E ^{(\cos \varphi - 1)\abs{\alpha }^2} \leq \E ^{\delta ^4 \abs{\alpha }^2/4!} \E ^{-\varphi ^2 \abs{\alpha }^2/2}.
$$
Therefore we have	
\beq \label{di:intineq}
	\int _{- \delta }^{\delta } \E ^{-\varphi ^2 \abs{\alpha }^2/2} d \varphi \leq \int _{- \delta }^{\delta } \E ^{(\cos \varphi - 1)\abs{\alpha }^2} d \varphi \leq \E ^{\delta ^4 \abs{\alpha }^2/4!} \int _{- \delta }^{\delta }\E ^{-\varphi ^2 \abs{\alpha }^2/2} d \varphi .
\eeq
Let $\delta = \abs{\alpha }^{-3/4}$.  Then, for $\alpha \rightarrow \infty$,
\beq \label{di:lim1}
   \E ^{\delta ^4 \abs{\alpha }^2/4!} = \E ^{\abs{\alpha }^{-3} \abs{\alpha }^2/4!} = \E ^{\abs{\alpha }^{-1} /4!} \rightarrow 1 
\eeq
and
\beq \label{di:gauss}
	\begin{aligned}
&{}&	\frac{\abs{\alpha }}{\sqrt{2 \pi }} \int _{- \delta }^{\delta } \E ^{-\varphi ^2 \abs{\alpha }^2/2} d \varphi 
	= \frac{1}{\sqrt{2 \pi }} \int _{- \delta \abs{\alpha }}^{\delta \abs{\alpha }} \E ^{-x^2/2} d x\\	
	&=& \frac{1}{\sqrt{2 \pi }} \int _{- \abs{\alpha }^{1/4}}^{\abs{\alpha }^{1/4}} \E ^{-x^2/2} d x
	\rightarrow \frac{1}{\sqrt{2 \pi }} \int _{- \infty }^{\infty } \E ^{-x^2/2} d x = 1 	
	\end{aligned}
\eeq
	hold as $\abs{\alpha } \rightarrow \infty $.
	Combining the relations (\ref{di:intineq}) - (\ref{di:gauss}) we get 
	$$
	\frac{\abs{\alpha }}{\sqrt{2 \pi }} \int _{- \delta }^{\delta } 
	\E ^{(\cos \varphi - 1)\abs{\alpha }^2} d \varphi 
	\rightarrow 1
	$$
	as $\abs{\alpha } \rightarrow \infty $.  This and the following relations show the desired relation (\ref{di:diracd}).
$$
\begin{aligned}
 0& \leq \int _{\delta }^{\pi } \E ^{(\cos \varphi - 1)\abs{\alpha }^2} d \varphi +
	\int _{-\pi }^{-\delta } \E ^{(\cos \varphi - 1)\abs{\alpha }^2} d \varphi \\ 
	&\leq  2\E ^{\delta ^4 \abs{\alpha }^2/4!} \int _{\delta }^{\pi }\E ^{-\varphi ^2 \abs{\alpha }^2/2} d \varphi\\ 
	&\leq  2\E ^{\delta ^4 \abs{\alpha }^2/4!} \E ^{-\delta ^2 \abs{\alpha }^2/2} (\pi - \delta ) 	\rightarrow 0  
\end{aligned}
$$
	for $0 < \delta \leq \pi $ and $\abs{\alpha } \rightarrow \infty $.
\end{proof}
\begin{cor} \label{di:cordirac}
	For a continuous function $f(\varphi )$ we have	
$$
	\frac{\abs{\alpha }}{\sqrt{2 \pi }} \int _{a - \pi }^{a + \pi } f(\varphi ) \E ^{(\cos (\varphi - a) - 1)\abs{\alpha }^2} d \varphi 
	\rightarrow f(a)
$$
as $\abs{\alpha } \rightarrow \infty $.	
\end{cor}
\begin{proof}
	By changing the variable $\varphi - a \rightarrow \varphi $, we have the corollary.
\end{proof}

Now we can prove the formula (\ref{di:formulamain}) for $x = 0$.

\begin{thm} \label{di:main1}
	Let $\Pi ^0$ be the projection operator onto the eigen-space of the observable $\Xi $ of (\ref{di:opXi}) with eigen-value $0$ and $\alpha = \E ^{\iu \theta } \abs{\alpha }$.  Then
	$$
	\lim _{\abs{\alpha } \rightarrow \infty }  \abs{\alpha } {}_2 \lan \alpha \vert \Pi ^0 \vert \alpha \ran _2
	= \frac{1}{\sqrt{2 \pi }} \E ^{ - \E ^{2 \iu \theta } \bar{u}_1^2/2}
	\overline{\E ^{ - \E ^{2 \iu \theta } \bar{v}_1^2/2}} ,
	$$
	where $\E ^{ - \E ^{2 \iu \theta } \bar{u}_1^2/2}$ is the eigen-function of the observable $\xi _1(\theta )$ of (\ref{di:opxitheta}) with eigen value $0$.
\end{thm}
\begin{proof}
	According to Proposition \ref{di:propcosl}, we calculate
	$$
	\lim _{\abs{\alpha } \rightarrow \infty } \frac{\abs{\alpha }}{2 \pi } \int _{a - \pi}^{a + \pi } 
	{}_2\lan \alpha \vert \varphi , 0 \ran \lan \varphi , 0 \vert \alpha \ran _2 d \varphi .
	$$	
	Let $\alpha = \E^{\iu \theta } \abs{\alpha }$. Then from the formula (\ref{di:alcos0}) we have
	\begin{align*}
	{}_2\lan \alpha \vert \varphi , 0 \ran \lan \varphi , 0 \vert \alpha \ran _2
&= \E^{- \vert \alpha \vert ^2/2} \E ^{\E ^{\iu \varphi }(- \bar{u}_1^2 + \E^{-2\iu \theta } \abs{\alpha }^2)/2} 
	\E^{- \vert \alpha \vert ^2/2} \E ^{\E ^{-\iu \varphi }(- v_1^2 + \E^{2\iu \theta } \abs{\alpha }^2)/2} 
	\\
	&= \E^{- \vert \alpha \vert ^2} \E ^{-(\E ^{\iu \varphi } \bar{u}_1^2/2 + \E ^{-\iu \varphi } v_1^2/2) + \cos (\varphi - 2\theta )\abs{\alpha }^2)} 
	\end{align*}
	and we get $\frac{1}{\sqrt{2 \pi }} \E ^{ - \E ^{2 \iu \theta } \bar{u}_1^2/2} \overline{\E ^{ - \E ^{2 \iu \theta } \bar{v}_1^2/2}}$	by using Corollary \ref{di:cordirac} {for $a = 2\theta $ and $f(\varphi ) = \E ^{-(\E ^{\iu \varphi } \bar{u}_1^2/2 + \E ^{-\iu \varphi } v_1^2/2)}$}.  This completes the proof.
\end{proof}

Next, we calculate $\lim _{\abs{\alpha } \rightarrow \infty }  \abs{\alpha } {}_2 \lan \alpha \vert \Pi ^l \vert \alpha \ran _2$ for $l = [x \abs{\alpha }]$, $x \neq 0$.
In the holomorphic representation, ${}_2\lan \alpha \vert \varphi , l \ran $ is expressed as
\beq \label{di:phipo}
  {}_2\lan \alpha \vert \varphi , l \ran = \E ^{- \abs{\alpha }^2/2} \left[ 
  \sum _{j=0}^{\infty } \E ^{\iu j \varphi } \frac{\sqrt{j!}}{\sqrt{2^l(l+j)!}} (\bar{u}_1 + \bar{\alpha })^l \frac{1}{2^j j!} (- \bar{u}_1^2 + \bar{\alpha }^2)^j \right] 
\eeq
for $l \geq 0$ and
$$
{}_2\lan \alpha \vert \varphi , l \ran = \E ^{- \abs{\alpha }^2/2} \left[ 
\sum _{j=0}^{\infty } \E ^{\iu j \varphi } \frac{\sqrt{j!}}{\sqrt{2^{\abs{l}}(\abs{l}+j)!}} (-\bar{u}_1 + \bar{\alpha })^{\abs{l}} \frac{1}{2^j j!} (- \bar{u}_1^2 + \bar{\alpha }^2)^j \right] 
$$
for $l \leq 0$.  We calculate ${}_2\lan \alpha \vert \varphi , l \ran $ only for the case of $l \geq 0$, because the calculation for $l \leq 0$ goes in the same way.
Note that (see (\ref{di:alcos0}))
\begin{align*}
{}_2\lan \alpha \vert \varphi , 0 \ran &= \E^{- \vert \alpha \vert ^2/2} \sum _{j=0}^{\infty } \E^{\iu j \varphi } \frac{1}{2^j j!} (- \bar{u}_1^2 + \bar{\alpha }^2)^j  
= \E^{- \vert \alpha \vert ^2/2} \E ^{\E ^{\iu \varphi }(- \bar{u}_1^2 + \E^{-2\iu \theta } \abs{\alpha }^2)/2} \\
&= \E ^{-\E ^{\iu \varphi } \bar{u}_1^2/2} \E^{- \vert \alpha \vert ^2/2}  
\sum_{j=0}^{\infty} \E ^{\iu j(\varphi - 2\theta )} \frac{(\abs{\alpha }^2/2)^j}{j!} .
\end{align*}
Comparing this with (\ref{di:phipo}), the terms $\frac{\sqrt{j!}}{\sqrt{2^l(l+j)!}} (\bar{u}_1 + \bar{\alpha })^l$ change ${}_2\lan \alpha \vert \varphi , 0 \ran $ into ${}_2\lan \alpha \vert \varphi , l \ran $, and we will study how these terms change $\E ^{- \E^{2\iu \theta } \bar{u}_1^2/2}$ into $\E ^{-x^2/4} \E ^{x \E ^{\iu \theta } \bar{u}_1} \E ^{ - \E ^{2 \iu \theta } \bar{u}_1^2/2}$.

{Let $l$ be an even positive integer.}  Note that the following change of the sum induces the interesting effect:
\beq \label{di:changsum}
  \E ^{- \abs{\alpha }^2/2} \sum _{j=0}^{\infty } \E^{\iu j \varphi } \frac{1}{2^j j!} (- \bar{u}_1^2 + \bar{\alpha }^2)^j
  \longrightarrow \E ^{- \abs{\alpha }^2/2} \sum _{j=l/2}^{\infty } \E^{\iu j \varphi } \frac{1}{2^j j!} (- \bar{u}_1^2 + \bar{\alpha }^2)^j
\eeq
\beq \label{in:changsum}
\begin{aligned} 
&= \E ^{- \abs{\alpha }^2/2} \sum _{j=0}^{\infty } \E^{\iu (j+l/2) \varphi } \frac{1}{2^{(j+l/2)} (j+l/2)!} (- \bar{u}_1^2 + \bar{\alpha }^2)^{(j+l/2)} 
\\
&= \E^{\iu l \varphi /2} \frac{1}{\sqrt{2^l}} (- \bar{u}_1^2 + \bar{\alpha }^2)^{l/2} \E ^{- \abs{\alpha }^2/2} \sum _{j=0}^{\infty } \E^{\iu j \varphi } \frac{j!}{(j+l/2)!} \frac{1}{2^j j!} (- \bar{u}_1^2 + \bar{\alpha }^2)^j .
\end{aligned}
\eeq
The multiplication {of each $j$-th term in (\ref{in:changsum})} with 
$$
\frac{(\bar{u}_1 + \bar{\alpha })^l}{(- \bar{u}_1^2 + \bar{\alpha }^2)^{l/2}}  \frac{(j + l/2)!}{\sqrt{(j+l)! j!}}
$$
induces the change
$$
\rightarrow \E^{\iu l \varphi /2} \frac{1}{\sqrt{2^l}} (\bar{u}_1 + \bar{\alpha })^l \E ^{- \abs{\alpha }^2/2} \sum _{j=0}^{\infty } \E^{\iu j \varphi } \frac{\sqrt{j!}}{\sqrt{(j+l)!}} \frac{1}{2^j j!} (- \bar{u}_1^2 + \bar{\alpha }^2)^j = \E^{\iu l \varphi /2} {}_2\lan \alpha \vert \varphi , l \ran .
$$
Equivalently,
\beq \label{di:psum}
  \frac{(\bar{u}_1 + \bar{\alpha })^l}{(- \bar{u}_1^2 + \bar{\alpha }^2)^{l/2}} \E^{- \vert \alpha \vert ^2/2} \sum _{j=l/2}^{\infty } \E^{\iu j \varphi } 
  \frac{j!}{\sqrt{(j+l/2)! (j - l/2)!}}
  \frac{1}{2^j j!} (- \bar{u}_1^2 + \bar{\alpha }^2)^j = \E^{\iu l \varphi /2} {}_2\lan \alpha \vert \varphi , l \ran .
\eeq
Now we study the formula (\ref{di:psum}).
\begin{lem} \label{di: cpsum}
	For any $\lambda , m > 0$ we have
$$
	\sum _{0 \leq j \leq m - \lambda \sqrt{m}} \E ^{-m} \frac{m^j}{j!} + \sum _{m + \lambda \sqrt{m} \leq j} \E ^{-m} \frac{m^j}{j!} \leq \frac{1}{\lambda ^2} .
$$
\end{lem}
\begin{proof}
The above inequality is nothing but the Chebyshev's inequality Pr$(\abs{X - \mu} \geq \lambda  \sigma ) \leq 1/\lambda ^2$ for the Poisson random variable $X$ with mean $\mu = m$ and standard deviation $\sigma = \sqrt{m}$.
\end{proof}

\begin{prop}
Let $M > 0$ and $0 < \theta < 1$.  Then
$$
	\E ^{-M} \sum _{0 \leq j \leq \theta M} \frac{M^j}{j!} \rightarrow 0
$$
as $M \rightarrow \infty $.
\end{prop}
\begin{proof}
For any $\lambda  > 0$, there exists $M > 0$ such that $\theta m \leq m - \lambda \sqrt{m}$ for $m \geq M$, and it follows from Lemma \ref{di: cpsum} that
$$
  \E ^{-m} \sum _{0 \leq j \leq \theta m} \frac{m^j}{j!} \leq \sum _{0 \leq j \leq m - \lambda \sqrt{m}} \E ^{-m} \frac{m^j}{j!} \leq \frac{1}{\lambda ^2} .
$$
This completes the proof.
\end{proof}
\begin{prop} \label{di:l/2->0}
Let $l/2 = [x \abs{\alpha }/2]$ for $x > 0$.  Then
$$
  \E^{- \vert \alpha \vert ^2/2} \sum _{j=0}^{l/2} \E^{\iu j \varphi } \frac{1}{2^j j!} (- \bar{u}_1^2 + \bar{\alpha }^2)^j  \rightarrow 0
$$	
as $\abs{\alpha } \rightarrow \infty $.
\end{prop}
\begin{proof}
Let $\mu = \abs{\bar{u}_1^2}/2$, $m = \abs{\alpha }^2/2$ and $M = \mu + m$, and note that 
$$
  l/2 = {[x \sqrt{2m}/2]} < (\mu + m)/2 = M/2 
$$
for sufficiently large $m$.  Then 
$$
  \E ^{-\mu } \E ^{-m} \abs{\sum _{j=0}^{l/2} \E^{\iu j \varphi } \frac{1}{2^j j!} (- \bar{u}_1^2 + \bar{\alpha }^2)^j}
  \leq \E ^{-M} \sum _{0 \leq j \leq M/2} \frac{M^j}{j!} \rightarrow 0, \ (m \rightarrow \infty )
$$
follows from the previous proposition.	
\end{proof}
\begin{cor} \label{di:replace}
{Let $l/2 = [x \abs{\alpha }/2]$ for $x > 0$.  Then}	
$$
  \E^{- \vert \alpha \vert ^2/2} \sum _{j=l/2}^{\infty } \E^{\iu j \varphi } \frac{1}{2^j j!} (- \bar{u}_1^2 + \bar{\alpha }^2)^j
$$
can be replaced by
$$
  \E ^{- \E ^{\iu \varphi } \bar{u}_1^2/2} \E^{- \vert \alpha \vert ^2/2} \sum _{j=l/2}^{\infty } \E^{\iu j \varphi } \frac{1}{2^j j!} \bar{\alpha }^{2j} 
$$
for sufficiently large $\abs{\alpha }$.	
\end{cor}
\begin{proof}
It follows from the previous proposition that
\begin{align*}
 {}& \E^{- \vert \alpha \vert ^2/2} \sum _{j=l/2}^{\infty } \E^{\iu j \varphi } \frac{1}{2^j j!} (- \bar{u}_1^2 + \bar{\alpha }^2)^j
   \approx \E^{- \vert \alpha \vert ^2/2} \sum _{j=0}^{\infty } \E^{\iu j \varphi } \frac{1}{2^j j!} (- \bar{u}_1^2 + \bar{\alpha }^2)^j
\\
  &= \E^{- \vert \alpha \vert ^2/2} \E ^{\E ^{\iu \varphi } (- \bar{u}_1^2 + \bar{\alpha }^2)/2}
  = \E^{- \vert \alpha \vert ^2/2} \E ^{- \E ^{\iu \varphi } \bar{u}_1^2/2} \E ^{\E ^{\iu \varphi } \bar{\alpha }^2/2}
  = \E ^{- \E ^{\iu \varphi } \bar{u}_1^2/2} \E^{- \vert \alpha \vert ^2/2} \sum _{j=0}^{\infty } \E^{\iu j \varphi } \frac{1}{2^j j!} \bar{\alpha }^{2j}\\
  &\approx \E ^{- \E ^{\iu \varphi } \bar{u}_1^2/2} \E^{- \vert \alpha \vert ^2/2} \sum _{j=l/2}^{\infty } \E^{\iu j \varphi } \frac{1}{2^j j!} \bar{\alpha }^{2j} .
\end{align*}
\end{proof}
For the proof of Proposition \ref{di:xebaru1} we need the following lemma.
\begin{lem}
For any $R > 0$, $\displaystyle \log \left( 1 + \frac{z}{m} \right) ^{ma}$	converges to $za$ as $m \rightarrow \infty $, uniformly on the sets $\{ a \in \C ; \abs{a} \leq R\} $ and $\{ z \in \C ; \abs{z} \leq R\} $.
\end{lem}
\begin{proof}
We begin with the following Taylor's expansion formula of $\log \left( 1 + z \right) $ for $\abs{z} < \rho < 1$:
$$
  \log \left( 1 + z \right)  = z + z^2 \frac{1}{2\pi \iu } \int _{\abs{\zeta } = \rho } \frac{1}{\zeta ^2 (\zeta - z)} \log \left( 1 + \zeta  \right) d \zeta  .
$$		
Let $0 < \rho < 1$ and $\abs{\zeta } = \rho $.  Then
$$
  \log (1 - \rho ) \leq \log \abs{1 + \zeta } \leq \log (1 + \rho ), \ \abs{{\rm arg} (1 + \zeta )} \leq \sin ^{-1} \rho 
$$
and
$$
  \abs{\log \left( 1 + \zeta \right)} \leq \abs{\log \abs{1 + \zeta }} + \abs{{\rm arg} (1 + \zeta )}
  \leq \abs{\log (1 - \rho )} + \sin ^{-1} \rho .
$$
Let $\abs{z} < \rho $.  Then $\abs{\zeta - z} \geq \rho - \abs{z}$ and
\begin{align*}
  \abs{z^2 \frac{1}{2\pi \iu } \int _{\abs{\zeta } = \rho } \frac{1}{\zeta ^2 (\zeta - z)} \log \left( 1 + \zeta \right) d \zeta }
 & \leq \frac{1}{2\pi } \int _{\abs{\zeta } = \rho } \frac{\abs{z}^2}{\rho ^2 (\rho - \abs{z})}
  (\abs{\log (1 - \rho )} + \sin ^{-1} \rho ) d \abs{\zeta }
\\
  &= \frac{\abs{z}^2}{\rho (\rho - \abs{z})} (\abs{\log (1 - \rho )} + \sin ^{-1} \rho ) .
\end{align*}
 For any $R > 0$, there exists $M_1$ such that for $m \geq M_1$, $\abs{z/m} < \rho < 1$ holds for $\abs{z} \leq R$.
Therefore, for any $\epsilon > 0$, $R > 0$ and $\abs{a}, \abs{z} \leq R$, there exists $M \geq M_1$ such that
\begin{align*}
&  \abs{\log \left( 1 + \frac{z}{m} \right) ^{ma} - ma \frac{z}{m}} =  \abs{ma \frac{z^2}{m^2} \frac{1}{2\pi \iu } \int _{\abs{\zeta } = \rho } \frac{1}{\zeta ^2 (\zeta - z/m)} \log \left( 1 + \zeta \right) d \zeta } 
\\
&  \leq \frac{\abs{a} \abs{z}^2}{m} \frac{1}{\rho (\rho - \abs{z/m})} (\abs{\log (1 - \rho )} + \sin ^{-1} \rho )
  < \epsilon 
\end{align*}
for $m \geq M$.  This completes the proof.
\end{proof}
\begin{prop} \label{di:xebaru1}
{Let $l$ be a positive even integer and $x > 0$.}
For $\alpha = \abs{\alpha } \E ^{\iu \theta }$ such that $l = x \abs{\alpha }$
$$
  \frac{(\bar{u}_1 + \bar{\alpha })^l}{(- \bar{u}_1^2 + \bar{\alpha }^2)^{l/2}} \rightarrow \E ^{x \E ^{\iu \theta }\bar{u}_1} 
$$
as $\abs{\alpha } \rightarrow \infty $.
\end{prop}
\begin{proof}
It follows from the above lemma that
$$
\frac{(\bar{u}_1 + \bar{\alpha })^l}{(- \bar{u}_1^2 + \bar{\alpha }^2)^{l/2}}
= \frac{(\bar{u}_1 + \bar{\alpha })^l}{(\bar{u}_1 + \bar{\alpha })^{l/2}(-\bar{u}_1 + \bar{\alpha })^{l/2}}
= \frac{(\bar{u}_1 + \bar{\alpha })^{l/2}}{(-\bar{u}_1 + \bar{\alpha })^{l/2}}
$$  
$$  
= \frac{(\E ^{\iu \theta } \bar{u}_1 + \abs{\alpha })^{l/2}}{(- \E ^{\iu \theta } \bar{u}_1 + \abs{\alpha })^{l/2}}
= \frac{(\E ^{\iu \theta } \bar{u}_1/\abs{\alpha } + 1)^{x \abs{\alpha }/2}}{(- \E ^{\iu \theta } \bar{u}_1/\abs{\alpha } + 1)^{x \abs{\alpha }/2}} 
\rightarrow \frac{\E ^{\E ^{\iu \theta } \bar{u}_1 x/2}}{\E ^{- \E ^{\iu \theta } \bar{u}_1 x/2}} = \E ^{x \E ^{\iu \theta }\bar{u}_1}
$$
as $\abs{\alpha } \rightarrow \infty $.
\end{proof}
\begin{prop} \label{di:e-x2}
Let $l$ be a positive integer and $\epsilon (l) = 0$ for even $l$ and $\epsilon (l) = 1/2$ for odd $l$.  For $j \geq l/2$, 
$$
  \frac{j!}{\sqrt{(j+(l/2 \pm \epsilon (l)))! (j - ((l/2) \pm \epsilon (l)))!}} \leq 1 .
$$
Let $m = \abs{\alpha }^2/2$, {$l/2 = [x \sqrt{2m}/2] + \epsilon $ and $j = [m + \lambda \sqrt{m}] + \delta = [(1 + \lambda /\sqrt{m})m] + \delta = [\mu m] + \delta $, where $\epsilon , \delta \in \{ 0, 1 \}$.}	 Then
$$
  \frac{j!}{\sqrt{(j+l/2)! (j - l/2)!}} \rightarrow \E ^{-x^2/4\mu }
$$	
as $m \rightarrow \infty $ uniformly on the set $\abs{x^2/\mu } \leq R$ for any $R > 0$.
\end{prop}
\begin{proof}
The first inequality follows from the inequalities.  {For even integers $l$,}
$$
  \frac{j!j! }{(j+l/2)! (j - l/2)!} = \frac{(j - l/2+1) \cdots j}{(j+1) \cdots (j+l/2)}
  = \frac{j - l/2+1}{j+1} \cdots \frac{j+1}{j+1+l/2} < 1 .
$$
{For odd integers $l$,
\begin{align*}
 {}& \frac{j!j! }{(j+l/2 \pm 1/2)! (j - l/2 \mp 1/2)!} = \frac{(j - l/2 \mp 1/2 +1) \cdots j}{(j+1) \cdots (j+l/2 \pm 1/2)}\\[2mm]
{}&  = \frac{j - l/2 \mp 1/2 +1}{j+1} \cdots \frac{j+1}{j+1+l/2 \pm 1/2} \leq 1 .
\end{align*}
}
Stirling's approximation $n! \sim \sqrt{2 \pi n} \{ n/\E \} ^n$ implies
\begin{align*}
&\frac{j!j! }{(j+l/2)! (j - l/2)!} \approx \frac{2 \pi j \{ j/\E \} ^{2j}}{\sqrt{2 \pi (j+l/2)} \{ (j+l/2)/\E \} ^{(j+l/2)} \sqrt{2 \pi (j-l/2)} \{ (j-l/2)/ \E \} ^{j-l/2}}
\\
&= \frac{2 \pi j j^{2j}}{\sqrt{2 \pi (j+l/2)} (j+l/2)^{(j+l/2)} \sqrt{2 \pi (j-l/2)} (j-l/2)^{j-l/2}}
\\
&= \frac{j}{\sqrt{j+l/2} \sqrt{j-l/2}}
\frac{j^{2j}}{(j+l/2)^{(j+l/2)} (j-l/2)^{j-l/2}} 
\\
&= \frac{j}{\sqrt{j+l/2} \sqrt{j-l/2}} \frac{1}{(1+l/2j)^j (1-l/2j)^j} \frac{(1-l/2j)^{l/2}}{(1+l/2j)^{l/2}} .
\end{align*}
The uniform convergence of the following three sequences, as $m\to \infty$,
{
\begin{align*}
 & \frac{j}{\sqrt{j+l/2} \sqrt{j-l/2}} = \frac{1}{\sqrt{1+l/2j} \sqrt{1-l/2j}}
\\
 & = \frac{1}{\sqrt{1+ ([x \sqrt{2m}/2] + \epsilon )/([\mu m] + \delta )} \sqrt{1- ([x \sqrt{2m}/2] + \epsilon )/([\mu m] + \delta )}} \\
 &  = \frac{1}{\sqrt{1 - ([x \sqrt{2m}/2] + \epsilon )^2/([\mu m] + \delta )^2}}
  \longrightarrow 1 ,
\end{align*}
\begin{align*}
 & \frac{(1-l/2j)^{l/2}}{(1+l/2j)^{l/2}} = \frac{(1- ([x \sqrt{2m}/2] + \epsilon )/([\mu m] + \delta ))^{[x \sqrt{2m}/2] + \epsilon }}{(1+ ([x \sqrt{2m}/2] + \epsilon )/([\mu m] + \delta ))^{[x \sqrt{2m}/2] + \epsilon }} \\
 & = \frac{(1- \{ ([x \sqrt{2m}/2] + \epsilon )^2/([\mu m] + \delta )\} / ([x \sqrt{2m}/2] + \epsilon ))^{[x \sqrt{2m}/2] + \epsilon }}{(1+ \{([x \sqrt{2m}/2] + \epsilon )^2/([\mu m] + \delta )\} / ([x \sqrt{2m}/2] + \epsilon ))^{[x \sqrt{2m}/2] + \epsilon }}
\\
 & \rightarrow \frac{\E ^{-x^2/2\mu }}{\E ^{x^2/2\mu}} = \E ^{-x^2/\mu },
\end{align*}
where we used
$$
  ([x \sqrt{2m}/2] + \epsilon )^2/([\mu m] + \delta ) \rightarrow x^2 / \mu 
$$
and
$$
\frac{1}{(1+l/2j)^j (1-l/2j)^j} = \frac{1}{(1 - ([x \sqrt{2m}/2] + \epsilon )^2/([\mu m] + \delta )^2)^{[\mu m] + \delta }}
$$
$$
 = \frac{1}{(1 - \{([x \sqrt{2m}/2] + \epsilon )^2/([\mu m] + \delta )\} /([\mu m] + \delta ))^{[\mu m] + \delta }}
\rightarrow \frac{1}{\E ^{- x^2/2\mu }} = \E ^{x^2/2\mu } 
$$
}
as $m \rightarrow \infty $ on the set $\abs{x^2/\mu } \leq R$ for any $R > 0$ implies the uniform convergence of 
$$
  \frac{j!}{\sqrt{(j+l/2)! (j - l/2)!}} \rightarrow \E ^{-x^2/4\mu } .
$$		
\end{proof}
\begin{rem} \label{di:choice}
If $l$ is an odd integer, $l/2$ is not an integer, so we can not use $l/2$ in the summation $\E ^{- \abs{\alpha }^2/2} \sum _{j=l/2}^{\infty } \E^{\iu j \varphi } \frac{1}{2^j j!} (- \bar{u}_1^2 + \bar{\alpha }^2)^j$ of (\ref{di:changsum}).  Instead of $l/2$ we must choose $[l/2]$ or $[l/2]+1$.  Proposition \ref{di:e-x2} shows that the term $\frac{j!}{\sqrt{(j+l/2)! (j - l/2)!}}$ in the formula (\ref{di:psum}) is not sensitive to the choice.  But the term $\frac{(\bar{u}_1 + \bar{\alpha })^l}{(- \bar{u}_1^2 + \bar{\alpha }^2)^{l/2}}$ in (\ref{di:psum}) is very sensitive to the choice, i.e., 
$$
  \frac{(\bar{u}_1 + \bar{\alpha })^l}{(- \bar{u}_1^2 + \bar{\alpha }^2)^{[l/2]}} \rightarrow \infty , \ \ 
  \frac{(\bar{u}_1 + \bar{\alpha })^l}{(- \bar{u}_1^2 + \bar{\alpha }^2)^{[l/2] + 1}} \rightarrow 0
$$
as $\alpha \rightarrow \infty $.  We choose their geometric mean
$$
  \sqrt{\frac{(\bar{u}_1 + \bar{\alpha })^l}{(- \bar{u}_1^2 + \bar{\alpha }^2)^{[l/2]}}
  \frac{(\bar{u}_1 + \bar{\alpha })^l}{(- \bar{u}_1^2 + \bar{\alpha }^2)^{[l/2] + 1}}}
  = \frac{(\bar{u}_1 + \bar{\alpha })^l}{(- \bar{u}_1^2 + \bar{\alpha }^2)^{l/2}},
$$
then it has a nonzero limit (see Proposition \ref{di:xebaru1}).
\end{rem}

Now we have the main proposition.
\begin{prop} \label{di:mainprop}
For $\alpha = \E ^{\iu \theta } \abs{\alpha }$ and $\abs{\alpha } = l/x$ with a fixed $0 \neq x \in \R$ we have
$$
	\E^{\iu l \varphi /2} {}_2\lan \alpha \vert \varphi , l \ran \approx
	\E ^{x \E ^{\iu \theta }\bar{u}_1} \E ^{- \E ^{\iu \varphi } \bar{u}_1^2/2} \E ^{-x^2/4} 
	\E ^{(\cos (\varphi - 2\theta ) - 1)\abs{\alpha }^2/2}\E ^{\iu (\sin (\varphi - 2\theta ))\abs{\alpha }^2/2}
$$
for sufficient large {integer} $l = x\abs{\alpha }$.
\end{prop}
\begin{proof}
	It follows from (\ref{di:psum}), Remark \ref{di:choice}, Proposition \ref{di:xebaru1}, Proposition \ref{di:e-x2} and Corollary \ref{di:replace} that
\begin{align*}
 \E^{\iu l \varphi /2} {}_2\lan \alpha \vert \varphi , l \ran &= \frac{(\bar{u}_1 + \bar{\alpha })^l}{(- \bar{u}_1^2 + \bar{\alpha }^2)^{l/2}} \E^{- \vert \alpha \vert ^2/2} \sum _{j=l/2}^{\infty } \E^{\iu j \varphi } 
\frac{j!}{\sqrt{(j+l/2)! (j - l/2)!}}
\frac{1}{2^j j!} (- \bar{u}_1^2 + \bar{\alpha }^2)^j 
\\
  &\approx  \E ^{x \E ^{\iu \theta }\bar{u}_1} 
  \E^{- \vert \alpha \vert ^2/2} \sum _{j=l/2}^{\infty } \E^{\iu j \varphi } 
  \E ^{-x^2/4\mu }
  \frac{1}{2^j j!} (- \bar{u}_1^2 + \bar{\alpha }^2)^j 
\\
&  \approx \E ^{x \E ^{\iu \theta }\bar{u}_1} \E ^{- \E ^{\iu \varphi } \bar{u}_1^2/2} \E^{- \vert \alpha \vert ^2/2} \sum _{j=l/2}^{\infty } \E^{\iu j \varphi } \E ^{-x^2/4\mu } \frac{1}{2^j j!} (\E ^{-2\iu \theta }\abs{\alpha }^2)^j .
\end{align*}
Let $m = \abs{\alpha }^2/2$.  Then it follows from Lemma \ref{di: cpsum} that the summation $\sum _{j=l/2}^{\infty }$ can be replaced by
$\sum _{j=m - \lambda \sqrt{m}}^{m + \lambda \sqrt{m}}$.  If $j \in [m - \lambda \sqrt{m}, m + \lambda \sqrt{m}]$, then $\mu = j/m \in [1 - \lambda /\sqrt{m}, 1 + \lambda /\sqrt{m}]$.  So, $\E ^{-x^2/4\mu }$ can be replaced by $\E ^{-x^2/4}$ for large $m$.  Thus we have
\begin{align*}
  \E^{\iu l \varphi /2} {}_2\lan \alpha \vert \varphi , l \ran &\approx \E ^{x \E ^{\iu \theta }\bar{u}_1} \E ^{- \E ^{\iu \varphi } \bar{u}_1^2/2} \E ^{-x^2/4} \E^{- \vert \alpha \vert ^2/2} \sum _{j=l/2}^{\infty } \E^{\iu j \varphi } \frac{1}{2^j j!} (\E ^{-2\iu \theta }\abs{\alpha }^2)^j
\\
  &\approx \E ^{x \E ^{\iu \theta }\bar{u}_1} \E ^{- \E ^{\iu \varphi } \bar{u}_1^2/2} \E ^{-x^2/4} \E^{- \vert \alpha \vert ^2/2} \sum _{j=0}^{\infty } \E^{\iu j (\varphi - 2\theta )} \frac{\abs{\alpha }^2/2}{j!}
\\
  &= \E ^{x \E ^{\iu \theta }\bar{u}_1} \E ^{- \E ^{\iu \varphi } \bar{u}_1^2/2} \E ^{-x^2/4} 
  \E ^{(\cos (\varphi - 2\theta ) - 1)\abs{\alpha }^2/2}\E ^{\iu (\sin (\varphi - 2\theta ))\abs{\alpha }^2/2} 
\end{align*}
for large $\abs{\alpha }$.  This completes the proof.
\end{proof}
Now our main result follows easily.
\begin{thm} \label{di:main2}
	Let $\Pi ^l$ be the projection operator onto the eigen-space of the observable $\Xi $ of (\ref{di:opXi}) with eigen-value $l$ and $\alpha = \E ^{\iu \theta } \abs{\alpha }$ and $\abs{\alpha } = l/x$.  Then
$$
	\lim _{l \rightarrow \infty } \abs{\alpha } {}_2 \lan \alpha \vert \Pi ^l \vert \alpha \ran _2
	= \frac{1}{\sqrt{2 \pi }} \E ^{x \E ^{\iu \theta } \bar{u}_1} \E ^{ - \E ^{2 \iu \theta } \bar{u}_1^2/2} \E ^{-x^2/4} 
	\overline{\E ^{x \E ^{\iu \theta } \bar{v}_1}\E ^{ - \E ^{2 \iu \theta } \bar{v}_1^2/2} \E ^{-x^2/4}}, 
$$
where $\E ^{x \E ^{\iu \theta } \bar{u}_1}\E ^{ - \E ^{2 \iu \theta } \bar{u}_1^2/2}$ is the eigen-function of the observable $\xi _1(\theta )$ of (\ref{di:opxitheta}) with eigen-value $x$.
\end{thm}
\begin{proof}
Proposition \ref{di:mainprop} shows
$$
	{}_2\lan \alpha \vert \varphi , l \ran \lan \varphi , l \vert \alpha \ran _2 \approx \E ^{x \E ^{\iu \theta } \bar{u}_1} \E ^{ - \E ^{\iu \varphi } \bar{u}_1^2/2} \E ^{-x^2/4} \overline{\E ^{x \E ^{\iu \theta } \bar{v}_1} \E ^{ - \E ^{\iu \varphi } \bar{v}_1^2/2} \E ^{-x^2/4}} 
	\E ^{(\cos (\varphi - 2\theta ) - 1)\abs{\alpha }^2}
$$
for sufficiently large $l = x \abs{\alpha }$.
Proposition \ref{di:propcosl} and Corollary \ref{di:cordirac} show the theorem.
\end{proof}
 
\section{Conclusion} \label{conclusion}
For the balanced homodyne detection of $p_1$, we prepare creation and annihilation  operators $c_4^*$  resp. $c_4$ of a local oscillator $\vert \gamma \ran _4$ ($\alpha = \E ^{\iu \theta } \abs{\alpha })$ and a half-beamsplitter
$$
\E ^{\iu H_{hbs}^{14}}, \  H_{hbs}^{14} = \iu(\pi /4)(c_{1}^{*} c_{4} - c_{1}c_{4}^{*}) . 
$$
Let
$$
  c_3^*c_0 + c_0^*c_3 = \sum _{l = -\infty }^{\infty } l \Pi _{03}^l \ \ {\rm and} \ \ 
   c_4^*c_1 + c_1^*c_4 = \sum _{k = -\infty }^{\infty } k \Pi _{14}^k . 
$$

Case A): If we make the measurement of $c_3^*c_0 + c_0^*c_3$ and $c_4^*c_1 + c_1^*c_4$ for the state $(\E ^{\iu H_{hbs}^{01}} \otimes I) \vert \Phi \ran \otimes \vert \alpha \ran _3 \otimes \vert \gamma \ran _4$ and get an outcome $(l, k) = (x_- \abs{\alpha }, p_+ \abs{\gamma })$ within an error $\epsilon (\abs{\alpha }, \abs{\gamma })$, i.e., $l \in ((x_- - \epsilon )\abs{\alpha }, (x_- + \epsilon )\abs{\alpha }] = (a, b]\abs{\alpha }$ and $k \in (p_+ - \epsilon , p_+ + \epsilon ]\abs{\gamma } = (c, d]\abs{\gamma }$, then the projective measurement changes the state  to 
\begin{align*}
  &\Pi _{03}^{(a, b]\abs{\alpha }} \Pi _{14}^{(c, d]\abs{\gamma }} (\E ^{\iu H_{hbs}^{01}} \otimes I) \vert \Phi \ran \otimes \vert \alpha \ran _3 \otimes \vert \gamma \ran _4 
\\
 & \approx (2 \epsilon )^2 \vert 0; x_- \ran _0
  {}_0 \lan 0; x_- \vert \otimes
  \vert \pi/2; p_+ \ran _1
  {}_1 \lan \pi/2; p_+ \vert (\E ^{\iu H_{hbs}^{01}} \otimes I) \vert \Phi \ran \otimes \vert \alpha \ran _3 \otimes \vert \gamma \ran _4 
\\
 & = (2 \epsilon )^2 \vert x_- \ran _0 \otimes \vert p_+ \ran _1 \{ {}_0 \lan x_- \vert \otimes {}_1 \lan p_+ \vert 
  (\E ^{\iu H_{hbs}^{01}} \otimes I) \vert \Phi \ran \} \otimes \vert \alpha \ran _3 \otimes \vert \gamma \ran _4 ,
\end{align*}
which follows from Proposition \ref{di:pitop} (see Remark \ref{di:bhomodynecoll}).  Note that Proposition \ref{di:pitop} follows from Theorem \ref{di:main2}.

Case B):  If we measures $N_3 - N_0$ and $N_4 - N_1$ for the state $\E ^{\iu H_{hbs}^{03}} \E ^{\iu H_{hbs}^{14}} (\E ^{\iu H_{hbs}^{01}} \otimes I) \vert \Phi \ran \otimes \vert \alpha \ran _3 \otimes \vert \gamma \ran _4$, then the collasped state is
$$
  (2\epsilon )^2 \E ^{\iu H_{hbs}^{03}} \E ^{\iu H_{hbs}^{14}} 
  \vert x_- \ran _0 \otimes \vert p_+ \ran _1 \{ {}_1 \lan x_- \vert \otimes {}_1 \lan p_+ \vert 
  (\E ^{\iu H_{hbs}^{01}} \otimes I) \vert \Phi \ran \} \otimes \vert \alpha \ran _3 \otimes \vert \gamma \ran _4 .
$$
The essential part of the teleportation is represented in the following relations:
\begin{align*}
 {}_0 \lan x_- \vert \otimes {}_1 \lan p_+ \vert \otimes I)(\E ^{\iu H_{hbs}^{01}} \otimes I \vert \Phi \ran )
&= \pi ^{-1/2} \sum _{k=0}^{\infty } {}_0 \lan k \vert D_0(x_- + \iu p_+ )^* \otimes {}_1 \lan k \vert \otimes I \vert \Phi \ran \\ \nonumber
&\rightarrow D_2(\alpha )^* \vert \psi \ran _2
\end{align*}
as $q \rightarrow 1$ for
$$
\vert \Phi \ran = \pi ^{1/2} \vert \psi \ran _0 \otimes \sum _{n=0}^{\infty } q^n \vert n \ran _1 \otimes \vert n \ran _2 \in \h _0 \hat{\otimes } \h _1 \hat{\otimes } \h _2, \ \abs{q} < 1 
$$
see (\ref{in:hbsxp}) and (\ref{pbmeasure}).  The essential part $\{ \cdots \}$ is the same in both cases A) and B), and so, the teleportation $ \vert \psi \ran _0 \rightarrow D_2(\alpha )^* \vert \psi \ran _2$ occurs in both cases.

Thus, Theorem \ref{di:main2} implies that in the framework of projective measurements, the balanced homodyne measurement causes 
quantum teleportation.

\nopagebreak


\begin{thebibliography}{1}
	\providecommand{\url}[1]{\normalfont{#1}}
	\providecommand{\urlprefix}{}
	
	\bibitem{Br90}
    Samuel L. Braunstein, 
    \newblock Homodyne statistics.
    \newblock {\em Phys. Rev. A.}, 42:474--481, 1990.
	
	\bibitem{Fu98}
	A. Furusawa, J.L. S{\o}rensen, S.L. Braunstein, C.A. Fuchs, H.J.~Kimble and E.S.
	Polzik.
	\newblock Unconditional {Q}uantum {T}eleportation.
	\newblock {\em Science}, 282:706--709, 1998.
	
	
	\bibitem{NB08}
    Nagamachi, S.; Br\"uning, E., 
	Quantum Teleportation and Holomorphic
	Representation of CCR, \emph{Open System \& Information Dynamics}
	2008, 15, 155--172.

	\bibitem{Ge11}
    T. Gerrits, S. Glancy and S.W. Nam, A balanced homodyne detector and local oscilloator shaping for measuring optical Schr\"odinger cat states, \emph{Proc. of SPIE}  2011, Vol. 8033, 80330X-1--80330X-7.

	\bibitem{SF18}
    Takahiro Serikawa and Akira Furusawa,  500 MHz resonant photo detector for high-quantum-efficiency, low-noise homodyne measurement, \emph{Review of Scientific Instruments}
    2018, 89, 063120-1--063120-8.


	\bibitem{BN19}
     Br\"uning, E.; Nagamachi, S.,
     On foundational aspects of optical quantum communication, \emph{Journal of Modern Optics}
     2019, 66, 1476--1490.

   \bibitem{Ma21}
   Takaya Matsuuura, Kento Maeda, Toshihiko Sasaki and Masato Koashi,  Finite-size security of continuous-variable quantum key distribution with digtital signal processing, \emph{Nature Communications}
   2021, 12:252.	

	\bibitem{Li21}
    Wen-Bo Liu, Chen-Long Li, Yuan-Mei Xie, Chen-Xun Weng, Jie Gu, Xiao-Yu Cao, Yu-Shuo Lu, Bing-Hong Li, Hua-Lei Yin and Zeng-Bing Chen, Homodyne Detection Quadrature Phase Shift Keying Continuous-Variable Quantum key Distribution with High Excess Noise Tolerance, \emph{PRX QUANTUM 2}  2021, 040334.


   \bibitem{FT22}
   Kosuke Fukui and Shuntaro Takeda, Building a large-scale quantum computer with continuous-variable optical technologies, \emph{Journal of Physics B: Atomic, Molecular and Optical Physics 55}
   2022, 012001.	
	
  \bibitem{FS80}
  L. D. Faddeev and A. A. Sla{v}nov, Gauge Field: Introduction to Quantum Theory, \emph{Benjamin/Cummings Publishing Co. Inc., London-Amsterdam-Sydney-Tokyo}, 1980

  \bibitem{BP02}
  H.P. Breuer and F. Petruccione, The theory of open quantum systems,
  \emph{Oxford University Press: Oxford New York}, 2002.

  \bibitem{Wm09} 
  H.M. Wiseman and G.J. Milburn, Quantum Measurement and Control,
  \emph{Cambridge University Press: The Edinburgh Building, Cambridge CB2 8RU}, UK, 2009.


\end{thebibliography}
\end{document}